\documentclass[a4paper,USenglish]{lipics-v2021}

\pdfoutput=1 

\title{Simulating Chirality: Solving Distance-$k$-Dispersion on an 1-Interval Connected Ring} 

 \author{Brati Mondal}{Department of Mathematics, Jadavpur University, Kolkata-700032, India}{bratim.math.rs@jadavpuruniversity.in}{https://orcid.org/0009-0001-3017-9924}
 {(Optional) author-specific funding acknowledgements}
 \author{Pritam Goswami}{Department of Computer Science, Sister Nivedita University, Kolkata-700156, India}{pgoswami.cs@gmail.com}{https://orcid.org/0000-0002-0546-3894}
 {[funding]}
  \author{Buddhadeb Sau}{Department of Mathematics, Jadavpur University, Kolkata-700032, India}{buddhadeb.sau@jadavpuruniversity.in}{https://orcid.org/0000-0001-7008-6135}
  {[funding]}



\nolinenumbers 

\EventEditors{John Q. Open and Joan R. Access}
\EventNoEds{2}
\EventLongTitle{42nd Conference on Very Important Topics (CVIT 2016)}
\EventShortTitle{CVIT 2016}
\EventAcronym{CVIT}
\EventYear{2016}
\EventDate{December 24--27, 2016}
\EventLocation{Little Whinging, United Kingdom}
\EventLogo{}
\SeriesVolume{42}
\ArticleNo{23}
\usepackage{amsfonts, amssymb, amsmath,pifont,float,caption,multicol}
\usepackage[ruled,vlined,linesnumbered]{algorithm2e}
\usepackage{multirow}
\usepackage{booktabs}
\usepackage{makeidx} 

\newcommand{\R}{\mathcal{R}}
\newcommand{\C}{\mathcal{C}}
\newcommand{\D}{\mathcal{D}}
\newcommand{\In}{\texttt{Init}}

\newcommand{\MergeII}{\texttt{Merge-II}}

\newcommand{\Osc}{\texttt{Oscillate}}
\newcommand{\PreMergeI}{\texttt{PreMerge-I}}
\newcommand{\round}{\texttt{Roundabout}}
\newcommand{\preRound}{\texttt{PreRoundabout}}
\newcommand{\MergeI}{\texttt{Merge-I}}
\newcommand{\mix}{\texttt{0-Merge-1}}
\newcolumntype{P}[1]{>{\raggedright\arraybackslash}p{#1}}

\begin{document}
\maketitle

\begin{abstract}
We study the \textit{Distance-$k$-Dispersion} (D-$k$-D) problem for synchronous mobile agents in a 1-interval connected ring network having $n$ nodes and with $l$ agents, where $3\le l \le \lfloor\frac{n}{k}\rfloor$, without the assumption of \textit{chirality} (a common sense of direction for the agents). This generalizes the classical dispersion problem by requiring that agents maintain a minimum distance of $k$ hops from each other, with the special case $k = 1$ corresponding to standard dispersion.

The contribution in this work is threefold.
Our first contribution is a novel method that enables agents to \textit{simulate chirality} using only local information and bounded memory. This technique demonstrates that chirality is not a fundamental requirement for coordination in this model.

Building on this, our second contribution partially resolves an open question posed by Agarwalla et al.\ (ICDCN 2018), who considered the same model (1-interval connected rings, synchronous agents, no chirality). We prove that D-$k$-D—and thus dispersion is solvable from \textit{arbitrary initial configurations} under these assumptions (excluding vertex permutation dynamism) for any size of the ring network which was earlier limited to only odd-sized ring or a ring of size four.

Finally, we present an algorithm for D-$k$-D in this setting that works in $O(ln)$ rounds, completing the constructive side of our result.

Altogether, our findings significantly extend the theoretical understanding of mobile agent coordination in dynamic networks and clarify the role of chirality in distributed computation.

\end{abstract}

\section{Introduction}

In recent years, mobile agents in distributed networks have emerged as a significant research area for solving fundamental coordination problems. These agents are generally simple, autonomous, and low-cost mobile units that can collaboratively solve tasks such as \textit{gathering}, \textit{exploration}, \textit{flocking}, and \textit{dispersion}.

This paper focuses on the \textit{dispersion} problem, first introduced in \cite{AM2018ICDCN}. In this problem, $l$ mobile agents are deployed on a graph with $n$ nodes, and the goal of the problem is to make the agents reposition themselves autonomously so that each node contains at most $\lceil \frac{l}{n} \rceil$ agents. Note that, when $l \geq n$, dispersion ensures that each node holds at most one agent. Dispersion is closely related to other well-known problems such as \textit{scattering}~\cite{BFES2011SCATTERINGGRID}, \textit{exploration}~\cite{LDFS2016ICDCS}, \textit{self-deployment}~\cite{EB20117DEPLOYMENT, SSNY2022ALMOSTUNIFORMDEPLOYMENT}, and \textit{load balancing}~\cite{LOADBALANCING12FOCS}, and has practical applications in many real world scenarios such as relocating self-driving vehicles to charging stations, deploying robots in hazardous environments, and more.

We study a generalization of this problem known as \textit{Distance-$k$-Dispersion} (D-$k$-D), introduced in~\cite{MGS2025DKDICDCN25}. The goal of this problem is to reposition agents such that the distance between any two agents on the graph is at least $k$. By definition, this problem inherits all practical applications of dispersion, while also broadening its scope. For instance, if each agent has a sensing radius of $k$ hops, solving D-$k$-D helps maximize sensing coverage across the network by minimizing the overlapping covering area.

In this work, we focus on solving the D-$k$-D problem on a dynamic ring network, specifically under the 1-interval connectivity model (i.e., at any particular round at most one edge of the ring might stay missing), with $l \geq 3$ agents starting from any arbitrary initial deployment, and without assuming chirality (i.e., no common agreement on clockwise or counterclockwise direction). It is assumed that the agents are capable of local communication (i.e., agents can communicate only when they are on the same node) and have bounded memory.

\subsection{Literature, Background, and Motivation}

Over the last decade, the dispersion problem has been studied extensively across a range of models and network topologies~\cite{AAMKS2018ICDCN, AM2018ICDCN, CKMS2023CALDAM, DBS2021CALDAM, GKM2024CALDAM, GMMP2024JPDC, KA2019ICDCN, KM2023NETYS, KMS2019ALGO, KMS2020ICDCN, KMS2020ICDCS, KMS2020WALCOM, KMS2022JPDC, KS2021OPODIS, MM2019TAMC, MMM2020ALGOSENSORS, MMM2021IPDPS, MMM2021TCS, SSKM2020SSS, SSNYT2024DISC}, owing to its practical relevance. Many real-world scenarios involving $l$ agents and $n$ spatially distributed resources can be modeled as a dispersion problem to minimize coordination time and system cost. For example, smart electric vehicles coordinating to reach available charging stations can be modeled as a dispersion task—driving to a station may take minutes, but recharging takes hours, so efficient coordination is crucial.

The Distance-$k$ Dispersion problem generalizes dispersion and has its own set of real-world applications. For instance, in environments where each agent can sense or affect a region within $k$ hops, D-$k$-D helps ensure maximum effective coverage by maintaining sufficient separation between agents.

In the context of ring networks, 
Agarwalla et al.~\cite{AAMKS2018ICDCN} studied dispersion in dynamic rings under two key types of network dynamism:
\begin{itemize}
    \item \textit{1-interval connectivity}, where in any round, at most one edge (chosen arbitrarily by the adversary) may be missing.
    \item \textit{Vertex permutation}, where the adversary can reorder nodes arbitrarily in each round.
\end{itemize}

Assuming local communication (i.e., agents can communicate only when co-located), the authors proposed several algorithms under different assumptions:
\begin{enumerate}
    \item With chirality and full visibility, they presented \textit{VP-1-Interval Chain} to solve dispersion.
    \item In achiral settings where agents are co-located with full visibility, first they used \textit{No-Chiral-Preprocess} to establish chirality, followed by \textit{VP-1-Interval Chain} to achieve dispersion. This approach requires agents to store the smallest ID, thus requiring memory.
    \item When agents are scattered, they proposed \textit{Achiral-Odd-VP-1-Interval-Chain} and \textit{Achiral-Even4-VP-1-Interval-Chain} to solve dispersion, but only for rings with an odd number of nodes or exactly four nodes, respectively. The latter also requires memory.
\end{enumerate}

These results naturally raise the following question:

\begin{center}
\textbf{Is chirality necessary to solve dispersion on a 1-interval connected ring with $n = 2n'$ nodes for $n' (> 2) \in \mathbb{N} $, under the same model with memory as in~\cite{AAMKS2018ICDCN}?}
\end{center}

Our primary goal in this paper is to address this question by exploring the achiral setting in greater depth under the same assumptions as in~\cite{AAMKS2018ICDCN} (excluding vertex permutation), while allowing agents to have memory. This investigation reveals important insights into the relative power of the model with and without chirality. Building on these findings, we further study the D-$k$-D problem in the same model, partially resolving the open question posed in~\cite{AAMKS2018ICDCN} about solving dispersion from arbitrary configurations on even-sized 1-interval connected rings (excluding the vertex permutation dynamism). A summary of our results is presented in the following subsection.

 \subsection{Our Contribution}
In this work, we study the problem of D-$k$-D on a 1-interval connected ring under the same bounded memory model described in \cite{AM2018ICDCN}. Here, we consider the 1-interval connectivity dynamism only, while excluding the vertex permutation. This is because, in vertex permutation dynamism, D-$k$-D  is not at all solvable if $k\ge 2$, as the adversary can always swap vertices to make the distance between two agents less than $k$. 
While our original objective was to solve the D-$k$-D problem from any arbitrary configuration under the bounded-memory model described in~\cite{AAMKS2018ICDCN}, excluding vertex permutation dynamism (hereafter referred to as $\mathcal{M}$), our investigation led to a deeper and unexpected insight:

\noindent\textbf{The model $\mathcal{M}$ without chirality is computationally equivalent to $\mathcal{M}$ with chirality}. 

To formally establish this equivalence, we introduce a novel algorithm, \textsc{Achiral-2-Chiral}, which enables a team of $l \ge 3$ agents, starting without any common sense of direction, to collectively establish chirality. The algorithm completes its execution within $O(ln)$ synchronous rounds, where $l$ is the number of agents and $n$ is the number of nodes of the ring. This result is significant, as it shows that agents can simulate the benefits of chirality using only local communication, full visibility, and bounded memory ($O(\log n)$ specifically), even in dynamic settings.

Building on this, we then address the original D-$k$-D problem. Once chirality is established, agents execute our second algorithm, \textsc{Dispersed} to achieve a dispersed configuration and then the third algorithm, \textsc{Dispersed-To-$k-$Dispersed}, which solves the D-$k$-D problem in dynamic rings and guarantees termination within $O(ln)$ rounds, where $n$ is the size of the ring.

Together, these results not only settle the open question posed in~\cite{AAMKS2018ICDCN} about solving dispersion (equivalent to solving D-$k$-D where $k=1$) in even-sized dynamic rings without chirality, under bounded memory and full visibility partially (only for 1-interval connectivity dynamism), but also introduce a powerful framework for transforming achiral systems into chiral ones within the same model, without any additional assumptions.
 The following Table~\ref{tab:1} provides a comparative overview of this work with previous literature. 

\begin{table}[h!]

\centering
\footnotesize
\renewcommand{\arraystretch}{1.3}
\begin{tabular}{|P{3cm}|P{4.5cm}|P{4.5cm}|}
\hline
\textbf{Aspect} & \textbf{Agawalla et al. \cite{AAMKS2018ICDCN} (ICDCN, 2018)} & \textbf{This work} \\
\hline
\textit{Problem} & \textbf{Dispersion ($k=1$)} & Distance-$k$-Dispersion for any $k$ such that $\lfloor\frac{n}{l}\rfloor\ge k\ge 1$ \\
\hline
\textit{Initial Chirality Agreement} & No & No \\
\hline
\textit{Achieves Chirality Agreement in between} & No & Yes \\
\hline
\textit{Dynamism of the ring} & 1-interval connectivity + vertex permutation & Only 1-interval connectivity; vertex permutation makes the problem impossible if $k \ge 2$ \\
\hline
\textit{Initial Configuration} & Any arbitrary configuration & Any arbitrary configuration \\
\hline
\textit{Size of ring} ($n$) & $n$ is either odd or four & $n$ can be any natural number \\
\hline
\textit{Time Complexity} & $O(n)$ & $O(ln)$, where $l$ is the number of agents \\
\hline
\end{tabular}
\caption{Comparison with Agawalla et al. (ICDCN 2018)}
\label{tab:1}
\end{table}

\section{Preliminaries}
In this section, we discuss the preliminaries, including the model assumed for agents, the scheduler, and other relevant components, along with some necessary definitions, that will be used throughout the paper. In the following subsection, we first describe the agent model and the ring structure. In the subsequent sections, several necessary terms and the problem under consideration are formally defined.
\subsection{Model and Problem Definition}
Let $\R=(V_{\R}, E_{\R})$ be a connected network where $|v_{\R}|= n$ and $\forall\  v \in V_{\R}, deg(v)=2$. We call $\R$ a ring network with $n$ vertices. Let $A=\{ r_1,r_2,\cdots r_l\}$ be a set of $l$ agents arbitrarily deployed on the nodes of $\R$, where $l \ge 3$. Each of the agents has the same computational capabilities with a bounded persistent memory of $O(\log n)$ bits. Each of the agents has a unique ID in the range of $[1,n^c]$ for some constant $c$. We denote the ID of an agent $r_i$ as $ID(r_i)$ and the $x-$th bit of the ID of $r_i$ from right as $ID_x(r_i)$. The number of bits used to store ID for each agent is $B$. Thus $B \approx O(\log n)$. The agents are autonomous (i.e., no central control), identical (i.e., physically indistinguishable), and homogeneous (i.e., execute the same algorithm). 
One or more agents can occupy the same node at a particular time. 
If at a particular routine, more than one agent is at a node then we call that node a \textit{multiplicity} at that time. Otherwise, if a node has only one agent at any particular round, that node is called \textit{singleton} at that time. The agents do not have any chirality agreement but their clockwise direction is consistent throughout the execution.

The agents operate in a synchronous setting, where time is divided into rounds of equal duration. Each round consists of the following four stages:
\begin{enumerate}
    \item \textbf{Edge Removal (by the adversary):} At the beginning of each round, the adversary may arbitrarily remove one edge from the ring $\R$, or leave the ring unchanged.
    \item \textbf{Look:} Each agent, with \textit{full visibility}, captures a snapshot of the entire ring, identifying singleton nodes, multiplicity nodes, unoccupied nodes, and the location of the missing edge, if any. It also accesses its memory and that of any co-located agents (Known as \textit{Local Communication}). This enables the agent to determine the exact number of agents at its current node—an ability known as \textit{local strong multiplicity detection}. For other nodes, however, the agent can only distinguish between singleton and multiplicity without knowing the exact agent count, referred to as \textit{global weak multiplicity detection}.
    \item \textbf{Compute:} Using the information collected during the Look phase, each agent runs the given algorithm to determine its next action.
    \item \textbf{Move:} Based on the output of the Compute phase, each agent updates its memory and moves to an adjacent node, as specified by the algorithm. The move is assumed to be \textit{instantaneous}, i.e., at any particular round, during the Look stage all agents are seen to be on nodes of the ring.   
\end{enumerate}
\textbf{Dynamism of the ring $\R$:} The ring $\R$ is an 1-interval connected ring. In a 1-interval connected ring at most one edge of $\R$ might stay missing in any particular round. 
We now proceed to define the problem formally in the subsection.
\begin{definition}[\textbf{Problem Definition}: Distance-$k$-Dispersion]\label{def: Problem Definition}
    Let $l$ agents ($3\le l \le \lfloor \frac{n}{k} \rfloor$) reside on an 1-interval connected ring $\mathcal{R}$ with $n$ nodes. Initially, the agents are placed arbitrarily on the nodes of the ring. The problem \textit{Distance-$k$-Dispersion} (D-$k$-D) asks to reposition the agents to achieve a configuration that satisfies the following: \\
(i) The occupied node contains exactly one agent.\\
(ii) The distance between any two occupied nodes is at least $k$.\\
(iii) There exists at least one pair of nodes whose distance is exactly $k$.
\end{definition}
We need to define some necessary terms as preliminaries. The definitions are in the following subsection. Some diagrams (Figure~\ref{fig:chain} and Figure~\ref{fig: klink}) are given for better visualization of the definitions.

\subsection{Definitions and Preliminaries}
\begin{definition}[Configuration at time $t$]\label{def:config}
    Configuration at time $t$ or, $\C(t)$ is the 3-tuple $(\R, e_t, f_t)$ where $\R$ is the ring network on which agents are deployed, $e_t\in \{ \varnothing \}\cup E_{\R}$ such that at $t$, edge $e_t$ is missing in $\R$, and $f_t: V_{\R}\rightarrow \mathbb{N}$ be a function such that $f_t(v)$ denotes the number of agents on $v$ at time $t$.
\end{definition}
In a configuration $\C(t)$, we call a vertex $v\in V_{\R}$ occupied if $f_t(v)\ne0$. An occupied vertex $v$ in a configuration $\C(t)$, is called \textit{singleton} if $f_t(v)=1$. $v$ is called \textit{multiplicity} if $f_t(v)>1$.
Also, in a configuration $\C(t)=(\R, e_t, f_t)$, the vertex of agent $r$ is denoted as $v_t(r)$. If the time $t$ is clear from context, we use the symbol $v(r)$. If $f_t(v)\le 1$ for all $v\in V_{\R}$ we call the configuration $\C(t)$ \textit{dispersed}. 

Henceforth, $CW$ and $CCW$ denote the clockwise and counterclockwise direction respectively. 
Let $u$ and $v$ be two vertices of a ring $\R$ and $\D\in \{CW, CCW\}$ be a direction. An \textit{arc from $u$ to $v$ in direction $\D$, denoted by $(u,v)_{\D}$}, is the subgraph induced by the set of all vertices in between $u$ and $v$ starting from $u$ in the direction $\D$ including $u$ and $v$. 
 We define the \textit{distance from $u$ to $v$ in direction $\D$} as the length of the arc $(u,v)_{\D}$ and denote it as $d_{\D}(u,v)$.
\begin{definition}[Adjacent occupied node pair]\label{def: adjacent occupied node pair}
Let $\C(t)=(\R,e_t,f_t)$ be a configuration at some time $t$. Two nodes, $u,v$ in $\R$ are said to form adjacent occupied node pairs in $\C(t)$ if there exists a direction $\D \in \{CW, CCW\}$ such that both of the following conditions hold in $\C(t)$.
\begin{enumerate}
    \item $f_t(u),f_t(v) \ne 0$
    \item $\forall$ $w\in (u,v)_{\D}$ where $w$ is not $u$ or $v$, $f_t(w)=0$ 
\end{enumerate}
    
\end{definition}

\begin{definition}[Chain]\label{def:chain}
    Let $\C(t)=(\R,e_t,f_t)$ be a configuration at time $t$. Let $(u,v)_{\D}$ be an arc for some $u,v \in V_{\R}$ and $\D \in \{CW,CCW\}$. The arc $(u,v)_{\D}$ is called a chain if all the following conditions hold in $\C(t)$.
    \begin{enumerate}
        \item $f_t(x) \ne 0 ~\forall x \in (u,v)_{\D}$
        \item If $x_u$ and $x_v$ be the adjacent nodes of $u$ and $v$ respectively such that $x_u, x_v \notin (u,v)_{\D}$ (if exists), then $f_t(x_u)=f_t(x_v)=0$.
    \end{enumerate}
    We denote this chain as  $C_{\D}(u,v)$
\end{definition}
We call $(u,v)_{\D}$ \textit{the arc of the chain}  $C_{\D}(u,v)$. A chain $C_{\D}(u,v)$ is called an $i-$\textit{chain} if $d_{\D}(u,v)=i$. Two chains $C_{\D}(u, v)$ and $C_{\D}(x,y)$ are called \textit{adjacent} if at least one of the two arcs, $(v,x)_{\D}$ or, $(y,u)_{\D}$ has no occupied node. 

In a configuration $\C(t)$, let $r$ be an agent in a 0-chain for some $t$. If $f_t(v(r))=1$, then we call the chain a \textit{singleton 0-chain}. Similarly if agents $r_1$ and $r_2$ are in same 1-chain in $\C(t)$, such that $v(r_1)\ne v(r_2)$ and $f_t(v(r_1))=f_t(v(r_2))=1$ then, we call such 1-chains \textit{Singleton 1-chain}. Let in an configuration $\C(t)=(\R,e_t,f_t)$, $C_{\D}(u,v)$ be a  chain for some $u, v \in V_{\R}$ and $\D\in \{CW,CCW\}$. Let $u'$ and $v'$ be nodes adjacent to $u$ and $v$ respectively such that $u',v'\notin (u,v)_{\D}$. We call the arc $(u',v')_{\D}$ the \textit{extended arc of chain $C_{\D}(u,v)$}.
\begin{definition}[Symmetric Configuration]
    Let $\C(t)=(\R,e_t,f_t)$ be a configuration at time $t$. If $e_t\ne \varnothing$ then we call $\C(t)$ symmetric if the following condition is true in $\C(t)$
    \begin{itemize}
        \item $\exists$ an $i-$chain $C_{\D}(u,v)$ such that $i$ is odd and $e_t=\overline{u'v'}$, $u',v' \in (u,v)_{\D}$ and $d_{\D}(u,u')=d_{\D}(v',v)$ where $d_{\D}(u,u') < d_{\D}(u,v')$
    \end{itemize}
\end{definition}
\begin{definition}[Asymmetric Configuration]
    Let $\C(t)=(\R,e_t,f_t)$ be a configuration at time $t$. If $e_t\ne \varnothing$ then we call $\C(t)$ asymmetric if the following condition is true in $\C(t)$
    \begin{itemize}
        \item $\exists$ an $i-$chain $C_{\D}(u,v)$ whose extended arc be $(u_1,v_1)_{\D}$. Let for $e_t=\overline{u'v'}$, $u',v' \in (u_1,v_1)_{\D}$ and $d_{\D}(u_1,u')\ne d_{\D}(v',v_1)$ where $d_{\D}(u,u') < d_{\D}(u,v')$
         
    \end{itemize}
\end{definition}

\begin{definition}[visible direction of a  chain]\label{def:visible direction}
    Let $C_{\D}(u,v)$ be an $i-$chain in a configuration $\C(t)=(\R,e_t,f_t)$ where $i\ge 1$. We say that the chain $C_{\D}(u,v)$ has a visible direction if for exactly one of $u$ and $v$, say $w$, $f_t(w)=1$ and for the vertex, say $w'$ other than $w$ among $u$ and $v$, $f_t(w')>1$.
\end{definition}
If for a chain $C_{\D}(u,v)$ with visible direction in configuration $\C(t)=(\R,e_t,f_t)$, $f_t(u)=1$ then we say the direction of the chain is $\D$ otherwise if $f_t(u)>1$ then the direction is $\D'$, where $\D'$ is the reverse direction of $\D$. The class of all $i-$chains where $i\ge 2$  and all visibly directed 1-chains is called the \textit{feasible class} and is denoted as $\mathcal{FC}$. On the other hand, the class of 0-chains and 1-chains without visible direction is called the \textit{non-feasible class} or $\mathcal{NFC}$.

\begin{figure}[H]
    \centering
    \includegraphics[scale=0.55]{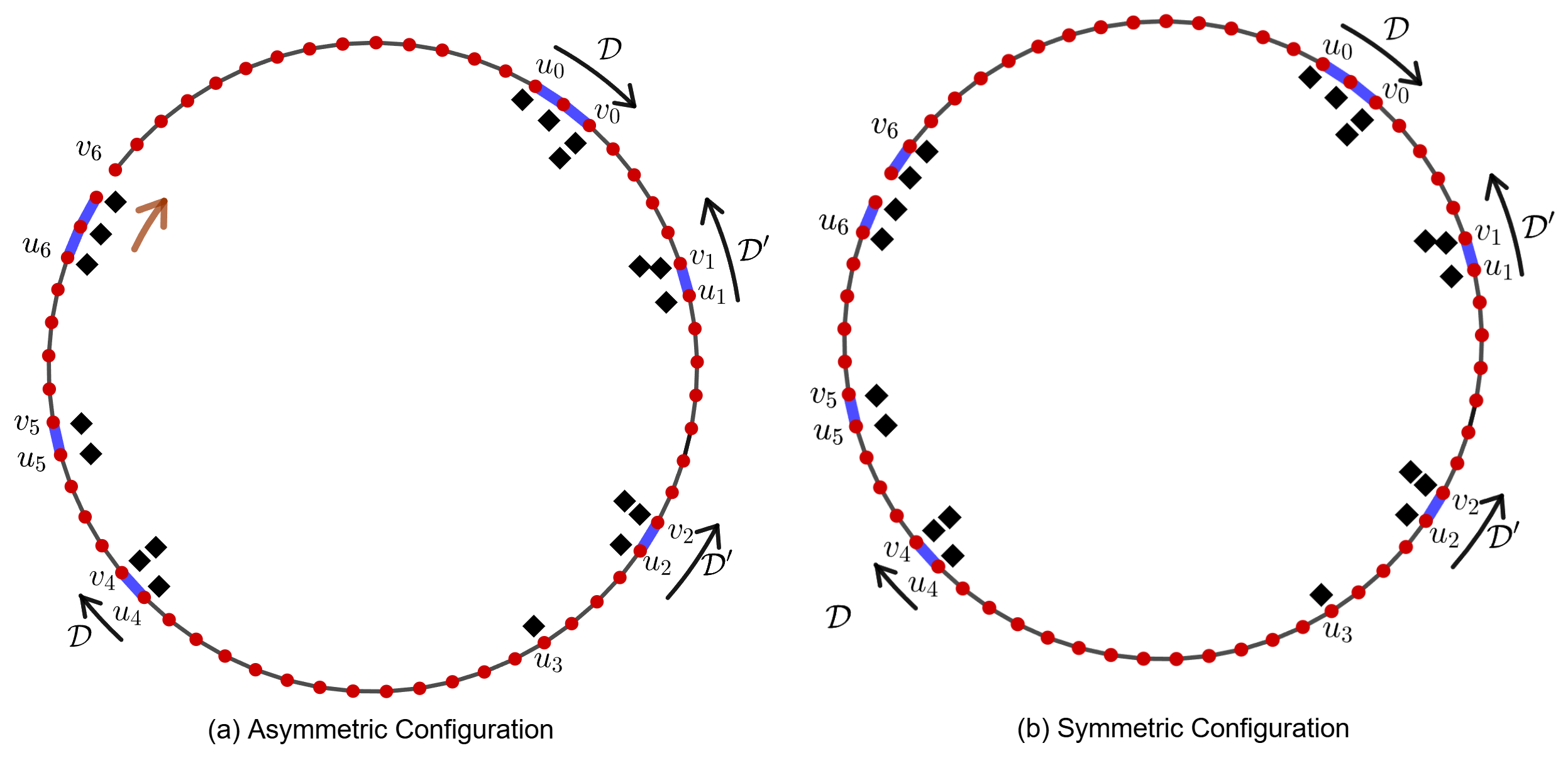}
    \caption{Here in this diagram arc of each chain are marked in blue. The visibly directed chains are marked with their corresponding direction. In each configuration the number of visibly directed chains in $\D$ and $\D'$ are same hence the configurations are not global. In the asymmetric configuration in left, the missing edge defines a global direction but in the symmetric configuration the missing edge is not defining any global direction.  }
    \label{fig:chain}
\end{figure}

\begin{definition}[Global Configuration]\label{def:global configuration}
    Let in a configuration $\C(t)$, there are $p$  chains with visible direction such that among these $p$ chains $p_1$ has a direction $\D$ for some $\D\in \{CW, CCW\}$ and rest $p_2=p-p_1$ chains has direction $\D'$ where $\D'$ is the reverse direction of $\D$. If $p_1\ne p_2$, then the configuration is referred to as a global configuration.  
\end{definition}
In a global configuration, the agents without chirality can agree on a global direction. If $p_1>p_2$ then the agents can agree on direction $\D$ as the global direction. Otherwise, if $p_2>p_1$, then the agents agree on $\D'$ as the global direction.

\begin{definition}[$k$-Link ]
Let $\mathcal{C}(t)=(\mathcal{R},e_t,f_t)$ be a dispersed configuration at time $t$. The agents $r_0,r_1,r_2, \dots, r_p$ are said to form a $k$-Link, in clockwise ($CW$) direction, if for any $i\in \{0,1,2,\dots,p-1\}$
$v(r_{i+1})$ is adjacent occupied node of $v(r_i)$ in direction $CW$ and the  distance $d_{CW}(v(r_i),v(r_{i+1}))< k$.
We denote this $k$-Link by $(r_0,r_1,\dots,r_p)_{CW}$. 
\end{definition}

\begin{definition}[Sub-$k$-Link of a $k$-Link]
    Let  there be two $k$-Links denoted as $KS_1=(b_0,b_1,\dots b_q)_{CW}$ and $KS_2=(a_0,a_1,\dots a_p)_{CW}$, we call $KS_1$ a sub-$k$-Link of $KS_2$ if both the following conditions hold.
    \begin{enumerate}
        \item $b_0=a_i$ for some $i\in \{0,1,\dots,p\}$ and $b_q=a_j$ for some $j\in \{0,1,\dots p\}$ where $i\le j$
        \item $b_x=a_{i+x}$ for any $0\le x\le p-i$
    \end{enumerate}
\end{definition}

We call a $k$-Link $KS=(a_0,a_1,\dots a_p)_{CW}$ \textit{Maximal} if $KS$ is not a sub-$k$-Link of any other $k$-Link. From this point onward, whenever we mention $k$-Link, we will mean maximal $k$-Link unless otherwise specified.
For a $k$-Link $KS=(r_0,r_1,\dots,r_p)_{CW}$, the node $v(r_0)$ is called the\textit{ tail } and the node $v(r_p)$ is called the \textit{head} of the $k$-Link $KS$ and are denoted as $H(KS)$ and $T(KS)$ respectively. By $|KS|$ we define the number of agents in the $k$-Link
\begin{definition}[Movable $k$-Link]
Let $KS$ be a $k$-Link in some configuration $\C(t)$ for some round $t$. $KS$ is called a movable $k$-Link if $d_{CW}(H(KS), T(KS'))>k$, where $KS'$ is the next $k$-Link in the clockwise direction of $KS$.
\end{definition}
Let $KS= KS_0$ be a movable $k$-Link. The \textit{Clockwise Nominee Set} and \textit{Counter clockwise Nominee Set} of $KS$ denoted by $\mathcal{NS}_{CW}(KS)$ and $\mathcal{NS}_{CCW}(KS)$ respectively are sets of agents defined as follows.
\begin{itemize}
    \item If $|KS|\ge 2$, then $\mathcal{NS}_{CW}(KS)=\{r \in A: r=H(KS)\}$ and $\mathcal{NS}_{CCW}(KS)=A \setminus \mathcal{NS}_{CW}(KS)$
    \item Let us consider the case where $|KS|=1$. Let all $k$-Links starting from $KS=KS_0$ in $CCW$ direction be $KS_0, KS_1, \cdots, KS_{\alpha}$. Let $KS_x$ be the first in the above order such that $|KS_x|\ge 2$. Now if $\nexists~ y \in \{1,2,\cdots x\}$ such that $KS_y$ is movable $k$-Link then $\mathcal{NS}_{CW}(KS)=\{r\in A : r=H(KS_j), 0\le j\le x\}$ and $\mathcal{NS}_{CCW}(KS)=A \setminus \mathcal{NS}_{CW}(KS) $ . Otherwise, let us consider $z$ (where $y\le z\le x$)  be such that $KS_z$ be  the last movable $k$-Link in the above mentioned order $\mathcal{NS}_{CW}(KS)= \{r\in A: r=H(KS_j), 0\le j\le y-1\}$. Then, $\mathcal{NS}_{CCW}(KS)=A \setminus \{r\in A: r=H(KS_j), z\le j \le x\}$.
\end{itemize}

\begin{figure}
    \centering
    \includegraphics[scale=0.017]{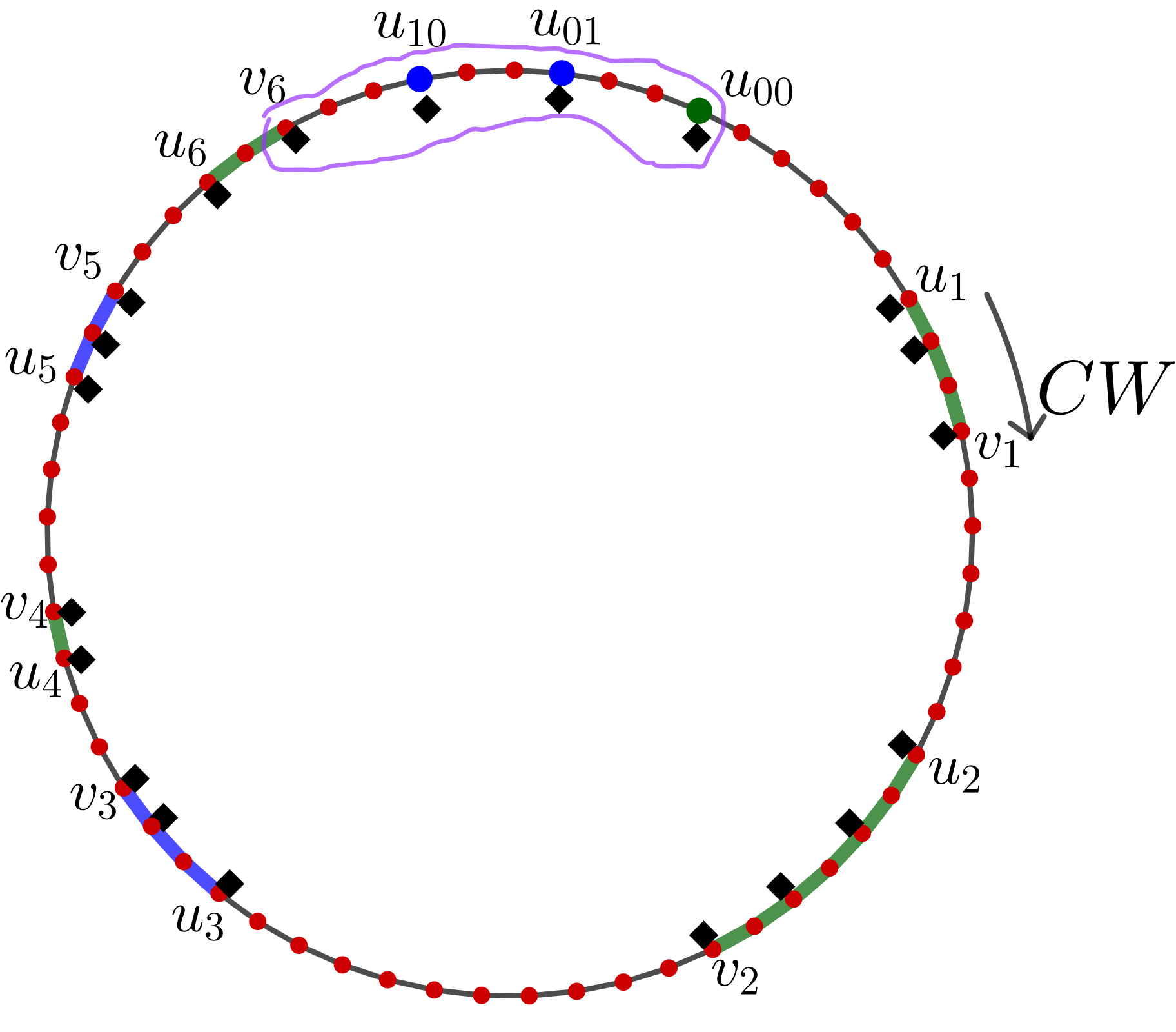}
    \caption{Assuming $k=3$, we have the movable $k$-Links which are marked with green arcs and the $k$-Links with blue arcs are non-movable $k$-Links. Let $KS$ be the $k$-Link with arc $(u_{00},u_{00})_{CW}$, then $\mathcal{NS}_{CW}(KS)=$ agents in arc $(v_6,u_{00})_{CW}$ enclosed in the purple loop and $\mathcal{NS}_{CCW}(KS)=$ rest of the agents}
    \label{fig: klink}
\end{figure}
\begin{definition}[Elected agent set in $\C(t)$] 
    Let $KS_0, KS_1, \cdots, KS_{\alpha}$ be all $k$-Links in $\C(t)$. Let $I=\{1,2, \cdots \alpha\}$ be the index set. Now let us consider the set $S(t)=\{KS_i: KS_i~ \text{is movable $k$-Link and}~ 0\le i \le \alpha\}$. Let $I' \subseteq I$ such that if $i\in I'$ then, $KS_i \in S(t)$. We define the \textit{elected agent set} in $\C(t)$, (denoted as $\mathcal{EAS}(t)$ ) to be $\mathcal{EAS}(t)= \underset{i\in I'}{\bigcup } \mathcal{NS}_{CW}(KS_i)$
\end{definition}

\section{Description of Algorithm \textsc{D-$k$-D\_Scattered\_Dynamic\_Ring}}

In this section, we describe the proposed algorithm \textsc{D-$k$-D\_Scattered\_Dynamic\_Ring} (Algorithm~\ref{D-k-D_scattered}) that solves the Distance-$k$-Dispersion problem defined in Definition~\ref{def: Problem Definition}. 

Initially, the agents do not have chirality. The algorithm works in three phases. In the first phase, the agents execute the algorithm \textsc{Achiral-2-Chiral} until chirality is achieved. It is guaranteed that, within at most $O(ln)$ rounds, all agents will agree on chirality. This novel algorithm can serve as a technique to simulate chirality in systems lacking an inherent sense of direction, under the assumptions of the considered model. After chirality is achieved, agents then execute Algorithm~\ref{Dispersed} (\textsc{Dispersed}) to achieve dispersion in the second phase. The agents execute this algorithm until the dispersion is achieved. Note that, at this point, the open problem regarding solving dispersion in any 1-interval connected ring from any arbitrary configuration is solved. To solve D-$k$-D from this dispersed configuration, the agents execute Algorithm~\ref{Algo:D-k-D} (\textsc{Dispersed To $k-$Dispersed}) until the problem is solved. The pseudocode of \textsc{D-$k$-D\_Scattered\_Dynamic\_Ring} is as follows.

\begin{algorithm}[!ht]
\caption{\textsc{D-$k$-D\_Scattered\_Dynamic ring}($r$,$\mathcal{C}(t_c)$}
\label{D-k-D_scattered}
\LinesNumbered
\uIf{$r$ has no global direction}
    {execute \textsc{Achiral-2-chiral}}
\Else
    {\uIf{in $\mathcal{C}(t_c)$, $\exists$ multiplicity node}
         {execute \textsc{Dispersed}($r$,$\mathcal{C}(t_c)$)}
     \Else
         {execute \textsc{Dispersed to $k$-Dispersed}($r$,$\mathcal{C}(t_c)$)}    
    }    
\end{algorithm}

In the following three subsections, we describe these three phases described above.

\subsection{Algorithm \mdseries{\textsc{Achiral-2-Chiral}}}
In this subsection, we describe the high-level idea of the algorithm \textsc{Achiral-2-Chiral} (Algorithm~\ref{algo:Main}). The main aim of this algorithm is to ensure chirality agreement for the agents within at most $O(ln)$ rounds, where $l(\ge 3)$ is the number of agents and $n$ is the number of nodes in the ring. This novel algorithm not only helps to solve D-$k$-D, but also works as an independent technique to achieve chirality from any achiral initial configuration under the concerned model. Thus, this algorithm proves that the concerned model $\mathcal{M}$ without chirality is computationally equivalent to $\mathcal{M}$ with chirality.  

\begin{algorithm}[!ht]
\caption{\textsc{Achiral-2-Chiral}}
\label{algo:Main}
\LinesNumbered
\eIf{$\C(t_c)$ is asymmetric}
{
    align the direction of asymmetry with clockwise direction\;
}
{
    \eIf{$\C(t_c)$ is global}
    {
        align the clockwise direction with global direction and terminate\;
    }
    {

        \uIf{state= \In}
        {
            execute \textsc{Guider}($r$, $\C(t_c)$)\;
        }
        \uElseIf{state= \Osc}
        {
            execute \textsc{GetDirected}($r$, $\C(t_c)$)\;
        }
        \uElseIf{state=\PreMergeI}
        {

            execute \textsc{PreProcess}($r$, $\C(t_c)$)\;
        }
        \uElseIf{state=\MergeI}
        {
            execute \textsc{Merge-I}($r$,$\C(t_c)$)\;
        }
        \uElseIf{state=\MergeII}
        {
            execute \textsc{Merge-II}($r$,$\C(t_c)$)\;
        }
        \uElseIf{state= \round}
        {
           
            execute \textsc{RoundTheRing-0}($r$,$\C(t_c)$)\;
        }
        \uElseIf{state=\preRound}
        {
            execute \textsc{RoundTheRing-I}($r$,$\C(t_c)$)\;
        }
        \uElseIf{state=\mix}
        {
            execute \textsc{Merge-0To1}($r$,$\C(t_c)$)
        }
    }
}
\end{algorithm}

There are three stages of execution of this algorithm. In the first stage, the goal of this algorithm is to ensure that the configuration becomes free of chains from the class $\mathcal{NFC}$. It is because the chains in the class $\mathcal{NFC}$ can't be made visibly directed, which is required to achieve a global direction agreement in the proposed algorithm. Next, the second and third stage repeats alternately until the chirality is achieved. In the second stage, all visibly undirected chains become visibly directed. So, after the completion of the second stage, all chains in the configuration become visibly directed. In this situation, if chirality is still not achieved, it means that the number of visibly directed chains in the clockwise direction is equal to the number in the counterclockwise direction. To resolve this, the agents proceed to the third stage, which ensures a reduction in the number of chains in the configuration. However, during this process, the newly formed chains may become visibly undirected. As a result, agents repeatedly alternate between Stage 2 and Stage 3 until the worst-case scenario, where there is exactly one visibly directed chain. In such a configuration, all agents agree on the direction of this chain as their global clockwise direction.

At any particular round, the agents can be in any one from the following set of states $\{\In, \round, \preRound, \PreMergeI, \MergeI, \MergeII, \Osc,\mix \}$. In the initial configuration, all agents are in state \In. If in any particular round, the configuration is global or asymmetric, then the agents can easily achieve chirality as described in the preliminaries. Otherwise, according to the algorithm, in any particular round, an agent executes one of the eight procedures depending upon the state it is in at that round. In the pseudocode of the algorithm \textsc{Achiral-2-Chiral}, we mention which procedure agents execute in a particular round under a specific state. The pseudocodes of the Procedures (excluding Procedure~\textsc{Guider()}) are not included in this part due to page limitations and can be found in the appendix.

\subsubsection{Overview of the algorithm}
Here we describe a high level idea of the algorithm \textsc{Achiral-2-Chiral}.
In $\C(0)$ all $l$ agents are placed arbitrarily in an 1-interval connected ring of size $n$. The agents do not have chirality and has a memory of $O(\log n)$.
Initially the agents are in state \In. 
During state \In, an agent execute Procedure~\ref{algo:Guider} (procedure \textsc{Guider}). This procedure guides the agents to different states depending on different conditions. The pseudo code is given below.

\begin{algorithm}[!ht]
\SetAlgorithmName{Procedure}{Procedure}{}
\caption{\textsc{Guider}($r$,$\C(t_c)$)}
\label{algo:Guider}
\LinesNumbered
    \eIf{$\C(t_c))$ has visibly directed  chains}
    {

            \eIf{all  chains are in $\mathcal{FC}$}
            {
               \eIf{all  chains are visibly directed}
               {
                    change state to \MergeII\;
               }
               {
                    Set $round=0$, $osc\_wait=0$ and $ret=0$\;
                    change state to \Osc\;
               }
            }
            {
                set $round=0$\;
                change state to \PreMergeI \;
            }
        
    }
    {
        \eIf{$\exists \  i-$chain where $i\le 1$}
        {
            set $round=0$\;
            change state to \PreMergeI\;
        }
        {
          Set $round=0$, $osc\_wait=0$ and $ret=0$\;
            change state to \Osc\;
        }
    }
\end{algorithm}
The main target of the algorithm is to create visibly directed  chains. If visibly directed  chains in direction $\D$ is greater than visibly directed  chains in direction $\D'$ then all agents align their clockwise direction with $\D$.

Now, it is easy to make an visibly undirected $i$-chain visibly directed if $i\ge 2$ (by performing Procedure~\ref{algo:GetDirected} (Procedure \textsc{GetDirected()}) which will be described in details later in this section). So the primary target of the algorithm is to eliminate all 0-chains and visibly undirected 1-chain by merging them with other  chains. We describe the elimination of 0-chain and visibly undirected 1-chains as follows.

Let the configuration is neither global nor asymmetric and contains 0-chains or visibly undirected 1-chains, for this the agents changes state to \PreMergeI~from \In (See Procedure~\ref{algo:Guider}). In this state agents execute Procedure~\ref{algo:preProcess} (Procedure \textsc{PreProcess}).

\begin{algorithm}[!ht]
\SetAlgorithmName{Procedure}{Procedure}{}
\caption{\textsc{PreProcess}($r$, $\C(t_c)$)} 
\label{algo:preProcess}
\LinesNumbered
\uIf{$round=0$}
{   $round++$\;
    \If{$r$ is an agent in a $0$-chain}
    {
        \If{$f_{t_c}(v(r))>1$ and $ID(r)$ is least on $v(r)$}
        {
            move in clockwise direction\;
        }
    }
}
\uElseIf{$round=1$}
{
    $round++$\;
    \If{$r$ is in a $1$-chain $\land$ both vertex of the  chain are multiplicity}
    {
        \If{$ID(r)$ is lowest in $v(r)$}
        {
            move outward of the  chain\;
        }
    }
}
\Else
{
    set $round=0$\;
    \eIf{In $\C(t_c)$ all chains are either Singleton $0$-chains or, Singleton $1$-chain}
    {
        
        \uIf{In $\C(t_c)$ all chains are Singleton $0$-chain}
         {   
            change state to \round\;
         }
         \uElseIf{In $\C(t_c)$ all chains are Singleton $1$-chain}
         {  
            set $move=0$\;
             change state to \preRound\;
         }
         \Else
         {
            change state to \mix\;
         }
         
    }
    {
        change state to \MergeI\;
    }

}
\end{algorithm}

After the first two round of executing this procedure it is ensured that there does not exists any 0-chain with multiplicity and visibly undirected 1- chain where both the occupied nodes of it are multiplicities. In this scenario there can be four cases.
\begin{itemize}
    \item []\textbf{Case-I:} All  chains are singleton 0-chains. In this case agents changes state to \round.
    \item []\textbf{Case-II:} All  chains are singleton 1-chains and all occupied nodes are singleton. In this case agents changes state to \preRound
    \item []\textbf{Case-III:} There exists 0-chain and 1-chains and all  chains are either 0-chain or 1-chain such that each occupied nodes are singleton. In this case agents changes state to \mix
    \item []\textbf{Case-IV:} There exists an $i-$chain $i \ge 2$ or visibly directed 1-chain. In this case agents changes state to \MergeI.
\end{itemize}

\begin{algorithm}[!ht]
\SetAlgorithmName{Procedure}{Procedure}{}
\caption{\textsc{RoundTheRing-0}($r$,$\C(t_c)$)} 
\label{algo:roundTheRing-0}
\LinesNumbered
\eIf{$round<n$}
{
    $round++$\;
    \If{$r$ is not a part of an $1$-chain $\lor$ $f_{t_c}(v(r))=1$}
    {
        move clockwise\;
    }
}
{
    \eIf{all  chains are Singleton $0$-chain}
    {
        align clockwise direction with global clockwise direction and terminate\;
        
    }
    {
        change state to \In\;
    }
}
\end{algorithm}

In state \round ~agents execute Procedure~\ref{algo:roundTheRing-0} ( Procedure \textsc{RoundTheRing-0}) for $n$ many rounds. To be specific agents move in their clockwise direction for $n$ many rounds unless they merge with another 0-chain. If after $n$ rounds all are still 0-chains then they have common clockwise sense so chirality is achieved else they move to a case similar to case II, III or, IV.

\begin{algorithm}[!ht]
\SetAlgorithmName{Procedure}{Procedure}{}
\caption{\textsc{RoundTheRing-I}($r$,$\C(t_c)$)} 
\label{algo:roundTheRing-1}
\LinesNumbered
\uIf{$round<B$}
{
$round++$\;
\eIf{$\C(t_c)$ is symmetric}
{
    \If{$r$ is a part of the 1-chain with $e_t$}
    {
        change state to \In\;
        move outwards \;
    }
}
{
    \If{$r$ is not part of a $0$-chain such that $f_{t_c}(v(r))>1$}
    {
        \If{$ID_{round+1}(r)=1$}
        {
            move inwards of the chain\;
        }
    }
}
    
}
\uElseIf{$round=B$}
{
    $round++$\;
    \tcp{\textcolor{blue}{In this point, all chains are 0- chains with 2 agents on each occupied node}}
    \If{$ID(r)$ is lowest in $v(r)$ }
    {
       set $move=1$\;
        move in its clockwise direction\;
    }
}
\uElseIf{$round=B+1$}
{
    $round++$\;
    \tcp{\textcolor{blue}{In this point, all chains are 1-chains with all occupied nodes singleton and one agent has $move=1$ the other has $move=0$ in each 1-chain}}
    \If{$move=0$}
    {
        align clockwise direction to the direction it sees the other agent in its  chain\;
        \tcp{\textcolor{blue}{At this point all agents in same chain agree in one particular direction}}
    }
}
\ElseIf{$B+1<round\le B+2n+1$}{
$round++$\;
    \eIf{$\C(t_c)$ is symmetric and $e_t$ is in the chain of $r$}
    {
        wait\;
    }
    {
       \If{$r$ is not a part of an $i-$chain where $i\ge 2$}
       {
            move clockwise\;
       }
    }
}
\Else
{
    \eIf{In $\C(t_c)$ all  chains are Singleton $1-$chain}
    {
        align clockwise direction with global clockwise direction and terminate\;
    }
    {
        change state to \In\;
    }
}
\end{algorithm}

In state \preRound, agents execute Procedure~\ref{algo:roundTheRing-1} ( Procedure \textsc{RoundTheRing-I}) for $2n + B + 2$ rounds, where $B$ is the bit-length of IDs. In the first $B+2$ rounds, if no symmetric configuration arises, agents in the same 1-chain agree on a clockwise direction as follows:

In round $t$ ($1 \le t \le B$), each agent checks the $t$-th bit of its ID (from the right). If it's 1, it moves inward within its chain; otherwise, it stays. If symmetry arises during this phase due to a missing edge $e_{t_c} = \overline{uv}$ in an 1-chain, agents on $u$ and $v$ move away, resulting in Case-III or Case-IV. On the other hand, if symmetry does not occur until the $B-$th round, all 1-chains reduce to 0-chains of two agents due to their unique Ids.
If no symmetry occurs by round $B$, then in round $B+1$, the agent with the lowest ID in each 0- chain moves clockwise. In round $B+2$, the other agent observes this movement and adopts the same direction, ensuring agreement within each chain.
For the next $2n$ rounds, agents follow their clockwise direction unless their chain includes the missing edge or has merged into a larger  chain. note that if after $2n$ rounds all of them are still singleton 1-chain then they agree on chirality as all of them has same clockwise direction. If one pair of chains has different clockwise sense then within $2n$ rounds, at least one chain moves $n$ times, confirming formation of a larger $i$- chain ($i \ge 2$). In this case agents again change state to \In~ leading to Case-IV.

\begin{algorithm}[!ht]
\SetAlgorithmName{Procedure}{Procedure}{}
\caption{\textsc{Merge-0To1}($r$,$\C(t_c)$)} 
\label{algo:merge-0to1}
\LinesNumbered
\eIf{$0$-chain exists $\land$ no chains are in $\mathcal{FC}$}
{
    
    \If{$r$ is part of $0$-chain and at least one adjacent  chain is 1- chain}
    {
        \eIf{exactly one  nearest adjacent 1-chain}
        {
            move in the direction of the 1-chain\;
        }
        {
            move clockwise\;
        }
    }
}
{
    change state to \In\;
}
\end{algorithm}

In state \mix, the agents execute Procedure~\ref{algo:merge-0to1} (Procedure \textsc{Merge-0To1}). This runs until the configuration contain at least one chain in $\mathcal{FC}$ or, the configuration contains no 0-chain. In this procedure an agent in a 0-chain moves in the direction of the nearest 1-chain if the 1-chain is adjacent to it. In case of tie (i.e agent in the 0-chain sees adjacent 1-chains in the same distance in both clockwise and counter clockwise direction) the agent moves in clockwise direction. This ensures after a finite time the configuration contains no 0-chains and there must exist at least one $i-$chain where $i \ge 2$ (Case-IV). 

\begin{algorithm}[!ht]
\SetAlgorithmName{Procedure}{Procedure}{}
\caption{ \textsc{Merge-I}($r$,$\C(t_c)$)} 
\label{algo:Merge-I}
\LinesNumbered
\eIf{$\exists$ chain in $\mathcal{NFC}$}
{
    \uIf{$r$ is in a $0$-chain $\land$ $f_{t_c}(v(r))=1$}
    {
        \uIf{(exactly one of adjacent  chain of $r$ is in $\mathcal{FC}$ in direction $\D$}
        {
            move in the direction $\D$
        }
        \uElseIf{Both adjacent  chains of $r$ are in  $\mathcal{FC}$}
        {
            move clockwise\;
        }
    }
    \ElseIf{$r$ is in a Singleton $1$-chain $C_{\D}(u,v)$}
    {
        \uIf{both adjacent chains of the chain of $C_{\D}(u,v)$ are in $\mathcal{FC}$}
            {
                 move outwards of the chain\;
            }
        \uElseIf{Exactly one adjacent chain $C_{\D}(u',v')$ of $C_{\D}(u,v)$ in  $\mathcal{FC}$ }
        {
            \If{$r$ is at $v$}
            {
                move in direction $\D$\;
            }
        }
       
    }
}
{
    change state to \In\;
}
\end{algorithm}

In state \MergeI, agents execute Procedure~\ref{algo:Merge-I} (Procedure \textsc{Merge-I}) until all 0-chains and visibly undirected 1-chains (i.e., all chains in $\mathcal{NFC}$) are eliminated. Note that for all occupied nodes in all 0-chains and visibly undirected 1-chains in state \MergeI~ are singleton. Since at least one $i-$chain ($i \ge 2$) exists, there exists at least one 0-chain or undirected 1-chain which is adjacent to such a chain or a visibly directed 1-chain. If a 0-chain is adjacent to exactly one such chain in direction $\D$, its agent moves in direction $\D$; if adjacent on both sides, it moves clockwise. This way the 0-chains will get merged in $i-$chains where $i\ge 1$ and makes a larger  chain. 

Now, if a singleton 1-chain is adjacent to an $i-$chain ($i \ge 2$) or a visibly directed 1-chain, its agents move outward. This move ensures that at least one moving agent among these two either gets merged to a larger chain or makes a singleton 0-chain  which will be further merged into a larger chain (using Procedure~\ref{algo:Merge-I}). This process ensures that, in finite time, all 0-chains and undirected 1-chains are eliminated.

\begin{algorithm}[!ht]
\SetAlgorithmName{Procedure}{Procedure}{}
\caption{\textsc{GetDirected}($r$, $\C(t_c)$)} 
\label{algo:GetDirected}
\uIf{$round=0$}
{
    $round++$\;
    \If{$v(r)$ is at terminal of an $i-$chain without visible direction and $i\ge 2$}
    {
        \If{$f_{t_c}(v(r))>1$ and $r$ is not the lowest ID agent on $v(r)$}
        {
            move inwards of the chain\;
        }
    }
}
\uElseIf{$1\le round \le 2B$}       
{
\tcp{\textcolor{blue}{$B$ is the ID length uniform for each agent as homogeneous}}
$round++$\;
\If{$v(r)$ is terminal node of an $i-$chain without visible direction where $i\ge 2$}
{
    \eIf{$round-1=2x+1$ for some $x<B$}
    {
        \If{$osc\_wait=0$}
        {
            \If{$ID_{x+1}(r)$=1}
            {
            \tcp{\textcolor{blue}{$ID_{x}(r)$ is the $x-$th bit of the ID of $r$ from right}}
                $ret=1$\;
                move inwards of the  chain\;
            }
        }
    }
    {
        \eIf{the  chain containing $v(r)$ is visibly directed}
        {
            $osc\_wait=1$\;
        }
        {
            \If{$ret=1$}
            {
                $ret=0$\;
                move outward of the chain\;
            }
        }
    }
}
}
\ElseIf{$ round > 2B$}
{change state to \In}
\end{algorithm}
At this point all agents are in state \In~and there are no 0-chains or visibly undirected 1-chains. Now if the configuration at this time is not asymmetric, the agents check if the configuration is global or not. If the configuration is not global then number of visibly directed  chains having one particular direction must be same with visibly directed 
 chains in the reverse direction (note that the count can be 0 too). In this scenario if there are visibly undirected  chains then agents changes state to \Osc~ and otherwise if all chains are visibly directed they changes state to \MergeII.

In state \Osc, the agents execute Procedure~\ref{algo:GetDirected} (Procedure \textsc{GetDirected()}). The main aim of this procedure is to make all  chains visibly directed. This procedure is executed by each agent for $2B+2$ rounds. In the first round (i.e., $round=0$) , the agents basically ensures that the terminal nodes of a visibly undirected  chain  becomes singleton. This is done by the agents located at a terminal node of a visibly undirected  chain. If the node is not singleton then all agents except the agent having the lowest ID at that node move inwards the  chain.  For the next $2B$ rounds, in each odd round $2x+1$ ($x < B$), singleton agents at terminal nodes check the $(x+1)$-th bit of their ID. If it is 1, they move inward; otherwise, they stay. In each even round $2x$ ($x \le B$), agents that moved in the previous round check if the chain is now visibly directed. If so, they wait; otherwise, they move outward. Note that, after $2B+1$ execution of procedure \textsc{GetDirected()}, all  chains becomes visibly directed due to the fact that there must exist one $x <B$ such that $ID_{x+1}(r_1)\ne ID_{x+1}(r_2)$, for any two agents $r_1$ and $r_2$ at the singleton terminal nodes of the same chain. So, at the $2B+2$-th round all chains become visibly directed and all agents changes state to \In.

\begin{algorithm}
\SetAlgorithmName{Procedure}{Procedure}{}
\caption{\textsc{Merge-II}($r$,$\C(t_c)$)} 
\label{algo:Merge-II}
\LinesNumbered

        \eIf{$\exists$ at least a pair of adjacent visibly directed chains such that directions of them are towards each other }
        {
            \If{$r$ is a part of such a  chain and $e_{t_c}$ is not a part of  chain of $r$ }
            {
                move towards the direction of such adjacent chain\;
            }
        }
        {
            change state to \In\;
        }


\end{algorithm}

Note that as described earlier, if the configuration is not global still now, then the agents change their state to \MergeII. At state \MergeII, the agents execute the Procedure~\ref{algo:Merge-II} (Procedure \textsc{Merge-II()}). The main aim of this procedure is to ensure that number of  chains decrease. Note that when all chains have visible direction and the configuration is not global then there must exist two adjacent  chains whose directions are facing each other. Using this fact, in this procedure, two adjacent  chains facing each other must move towards each other until they become a single  chain. Note that if one such  chain contain a missing edge in a symmetric configuration at some round then in that round no agent of that  chain moves. Since in each round one of the two adjacent chain can move they will eventually merge together. Note that this procedure is executed  until there is no more adjacent visibly directed  chains with their direction towards each other. At this stage the agents change state to \In~ again. 

Due to merging, the number of chains decreases. However, merging two visibly directed chains results in a longer visibly undirected chain. If the configuration is still not global, agents are guided back to state \Osc~ via \textsc{Guider()}, and re-execute \textsc{GetDirected()} to ensure visible direction for all  chains. This process repeats until, in the worst case, only one visibly directed  chain remains, making the configuration global. This way the agents can achieve chirality.

\subsubsection{Correctness and Analysis of Algorithm \mdseries{\textsc{Achiral-2-Chiral}}}
\label{App B}
Before proving the correctness of Algorithm \textsc{Achiral-2-Chiral}, let us first define the necessary notations for the proofs.
We define $\mathcal{Z}_0(t)$ and $\mathcal{Z}_1(t)$ to be the number of 0-chains and visibly undirected 1-chains in $\C(t)$ respectively. We define $\Phi(t)$ as $\Phi(t)= \mathcal{Z}_0(t)+\mathcal{Z}_1(t)$.
We denote the set of all  chains in $\C(t)$ that are in $\mathcal{NFC}$ as $S_{\Phi}(t)$. and the set of rest of the  chains in $\C(t)$ that are in $\mathcal{FC}$ as $S_{\Phi}'(t)$

\begin{lemma}
\label{lemma:merge1 decrease unwanted chains}
   Let $\C(t_0)$ be a configuration where all agents are in state \MergeI~ and $\Phi(t_0)>0$. Then, there exists a chain $C_{\D}(u,v) \in \mathcal{NFC}$ and a round $t_f>t_0$ such that in configuration $\C(t_f)$, all agents that were part of the chain $C_{\D}(u,v)$ have merged and become part of chains in the class $\mathcal{FC}$.
\end{lemma}

\begin{proof}
     Since in $\C(t_0)$, the agents are in state \MergeI, there must exist at least one chain from the class $\mathcal{FC}$. Also, it is given that at least one chain from class $\mathcal{NFC}$ exists in $\C(t_0)$. Then there must exists an adjacent pair of chains $C_{\D}(u_1,v_1) \in \mathcal{NFC}$ and $C_{\D}(u_2,v_2)\in \mathcal{FC}$. First note that, $C_{\D}(u_1,v_1)$ is either a singleton 0-chain or a singleton 1-chain as agents can have state \MergeI~ only after they execute the first two rounds of procedure~\ref{algo:preProcess}, which ensures the claim. Now there can be two cases, either $C_{\D}(u_1,v_1)$ is singleton 0-chain or singleton 1-chain.
     
     \textbf{Case-I:} In this case, let us  consider $C_{\D}(u_1,v_1)$ is a 0-chain (in this case $u_1=v_1$). Let $C_{\D}(u_2,v_2) \in \mathcal{FC}$ is adjacent to $C_{\D}(u_1,u_1)$. Let another adjacent chain of $C_{\D}(u_1,v_1)$ be $C_{\D'}(u_3,v_3)$ (if exists). 
     Let us first consider the case where $C_{\D'}(u_3,v_3)\notin \mathcal{FC}$. In this case the agent on $u_1$ in $\C(t_0)$ moves in direction $\D$ (following Procedure~\ref{algo:Merge-I}) until it reaches the adjacent node of $u_2$, say $u_2'$ that is in $(u_1,u_2)_{\D}$ at a round, say $t_1>t_0$. Note that, the chain $C_{\D}(u_2,v_2)$ in $\C(t_0)$ changes to a 2-chain $C_{\D}(u_2',v_2)$ in $\C(t_1)$. So, for this case taking $t_f=t_1$ and $C_{\D}(u,v)= C_{\D}(u_1,v_1)$ suffices our purpose. Now, for the case where $C_{\D'}(u_3,v_3)\in \mathcal{FC}$, without loss of generality let us assume the agent in $u_1$ moves towards $C_{\D}(u_2,v_2)$. Now similarly the agent merges with $C_{\D}(u_2,v_2)$ at some time $t_2>t_0$ creating the new 2-chain $C_{\D}(u_2',v_2)$ in $\C(t_2)$. For this case, assuming $t_f=t_2$ and  $C_{\D}(u,v)= C_{\D}(u_1,v_1)$, suffices our purpose. Note, here the adversary can not stop the movements of the agent in $u_1$ without the agents agreeing on chirality. Because, stopping the agent to move from a singleton 0-chain by removing an edge would make the configuration asymmetric, which ensures that the agents will get an agreement on chirality.

     \textbf{Case-II:} In this case, let us consider $C_{\D}(u_1,v_1)$ is a singleton 1-chain (in this case $d_{\D}(u_1, v_1)=1$). Let $C_{\D}(u_2,v_2) \in \mathcal{FC}$ be adjacent to $C_{\D}(u_1,v_1)$. Let $r_1$ and $r_2$ be two agents at $u_1$ and $v_1$ respectively in $\C(t_0)$.
     Let $C_{\D'}(u_1',v_1')$ be another adjacent chain of $C_{\D}(u_1,v_1)$ in $\C(t_0)$. Now there can be two subcases.

     \textit{\textbf{1.} $C_{\D'}(u_1',v_1') \notin \mathcal{FC}$ in $\C(t_0)$\textbf{:} }
     Let $C_{\D'}(u_3,v_3)$ be the last chain from the class $\mathcal{NFC}$ in direction $\D'$ from chain $C_{\D}(u_1,v_1S)$ and $C_{\D'}(u_4,v_4) \in \mathcal{FC}$ be adjacent to $C_{\D'}(u_3,v_3)$ in direction $\D'$. If we can prove that all agents of at least one chain among $C_{\D}(u_1,v_1)$ and $C_{\D'}(u_3,v_3)$ in $\C(t_0)$, merges with chains among $C_{\D}(u_2,v_2)$ and $C_{\D'}(u_4,v_4)$ then we are done with this case. For the sake of contradiction, let us assume there exists agents, say $r_1'$, and $r_2'$ in chains $C_{\D}(u_1,v_1)$ and $C_{\D'}(u_3,v_3)$ respectively  (in $\C(t_0)$) that do not merge with any of $C_{\D}(u_2,v_2)$ or $C_{\D'}(u_4,v_4)$. Note that $r_1'$ and $r_2'$, can not be at $v_1$ and $v_3$ in $\C(t_0)$ respectively. This is because, otherwise according to Procedure~\ref{algo:Merge-I} agent $r_1'$ and $r_2'$ first move in direction $\D$ and $\D'$ respectively and forms singleton 0-chains adjacent to  $C_{\D}(u_2,v_2)$ or, $C_{\D'}(u_4,v_4)$ respectively in $\C(t_0+1)$ (if not already merged with $C_{\D}(u_2,v_2)$ or $C_{\D'}(u_4,v_4)$). Again, these agents execute Procedure~\ref{algo:Merge-I} and $r_1'$ starts moving in direction $\D$ and $r_2'$ in direction $\D'$. Without loss of generality let us assume that $r_1'$ be the first agent that reaches $u_2'$ before $r_2'$ reaches $u_4'$ where, $u_2'$ is the vertex adjacent to $u_2$ in the arc $(v_1,u_2)_{\D}$ and $u_4'$ is adjacent to $u_4$ in $(v_3,u_4)_{\D'}$. This implies $r_1'$ merges with $C_{\D}(u_2,v_2)$ and becomes a part of the chain in $\mathcal{FC}$, which contradicts our assumption. So, now let $r_1'$ and $r_2'$ is at $u_1$ and $u_3$. As described in the procedure~\ref{algo:Merge-I}, there must exists a time $t_3>t_0$ when both agents those were on $v_1$ and $v_3$ merges with the chains $C_{\D}(u_2,v_2)$ and $C_{\D'}(u_4,v_4)$ respectively. So, in $\C(t_3)$, $r_1'$ forms a singleton 0-chain that is adjacent to the chain $C_{\D}(u2',v_2)$ and $r_2'$ also forms a singleton 0-chain that is adjacent to the chain $C_{\D'}(u_4',v_4)$ where, $u_2'$ and $u_4'$ are adjacent vertices of $u_2$ and $u_4$ respectively which are in the arcs $(v_1,u_2)_{\D}$ and $(v_3,u_4)_{\D'}$ respectively. 
     Let without loss of generality, $d_{\D}(u_1,u_2')\le d_{\D'}(u_3,u_4')$. Then in $\C(t_3+d_{\D}(u_1,u_2')-1)$, $r_1'$ reaches $u_2''$, where $u_2''$ is adjacent to $u_2'$ in direction $\D'$ from $u_2'$. This implies $r_1'$ merges with the chain $C_{\D}(u_2,v_2)$ and forms $C_{\D}(u_2'',v_2) \in \mathcal{FC}$. This contradicts our assumption. And so, the claim is true. 

     \textit{\textbf{2.} $C_{\D'}(u_1',v_1') \in \mathcal{FC}$ in $\C(t_0)$\textbf{:} } First we consider the distances $d_{\D}(v_1,u_2)$ and $d_{\D'}(u_1,u_1')$.
     First let us assume that $d_{\D}(v_1,u_2), d_{\D'}(u_1,u_1')\ge 3$. In this case, while executing Procedure~\ref{algo:Merge-I}, $r_1$ and $r_2$ first move away from each other and create two 0-chains. After that $r_1$ moves towards $u_1'$ in direction $\D'$ an $r_2$ moves in direction $\D$ until they reach $u_2'$ and $(u_1')'$ respectively where $u_2'$ is adjacent to $u_2$ in direction $\D'$ and $(u_1')'$ is adjacent to $u_1'$ in direction $\D$. This way both agents of chain $C_{\D}(u_1,v_1)$ from configuration $\C(t_0)$ merges with chains among $C_{\D}(u_2,v_2)$ and $C_{\D'}(u_1.v_1')$ in $\C(t_1)$ where $t_1=t_0+\max\{d_{\D}(v_1,u_2'), d_{\D'}(u_1,(u_1')')\}$. Now let exactly one among $d_{\D}(v_1,u_2)$ and  $d_{\D'}(u_1,u_1')$ be greater than 3 and the other is exactly 2. Without loss of generality, let $d_{\D}(v_1,u_2)\ge 3$ and  $d_{\D'}(u_1,u_1')=2$. In this case, as described earlier in $\C(t_0)$ both $r_1$ and $r_2$ move in the opposite of each other. This move ensures that $r_1$ merges with $C_{\D'}(u_1',v_1')$ after moving in $\C(t_0)$ and forms a new chain $C_{\D'}((u_1')',v_1') \in \mathcal{FC}$ in $\C(t_0+1)$ and $r_2$ in $\C(t_0+1)$ creates a 0- chain which is adjacent to $C_{\D}(u_2,v_2)\in \mathcal{FC}$ in direction $\D$ and also to $C_{\D'}((u_1')', v_1')$ in direction $\D'$. Let without loss of generality according to Procedure~\ref{algo:Merge-I}, $r_2$ moves in Direction $\D$ in $\C(t_0+1)$. Now in $C(t_0+d_{\D}(v_1,u_2))$, $r_2$ also merges with the chain $C_{\D}(u_2,v_2)$ from the configuration $\C(t_0)$ creating a new chain $C_{\D}(u_2',v_2) \in \mathcal{FC}$, where $u_2'$ is adjacent to $u_2$ in direction $\D'$. So for these cases also all agents of chain $C_{\D}(u_1,v_1)$ from configuration $\C(t_0)$ merge and form chains in class $\mathcal{FC}$. Now if both  $d_{\D}(v_1,u_2)=d_{\D'}(u_1,u_1')=2$, then in $\C(t_0+1)$, $r_1$ merges and creates chain $C_{\D'}((u_1')',v_1') \in \mathcal{FC}$ and $r_2$ 
     merges and creates chain $C_{\D'}(u_2',v_2) \in \mathcal{FC}$. 
     Note that, the adversary can not stop any move by any agent when all of them are in state \MergeI~ by removing edges without agents achieving chirality agreement. This is because, to do that adversary has to make the configuration asymmetric, and as a result of that, agents will achieve chirality.
     
     So, for any configuration $\C(t_0)$ for any $t_0$, if agents are in state \MergeI~ and $\Phi(t_0) > 0$ then we can see that there exists one chain from the class $\mathcal{NFC}$ in $\C(t_0)$ and $t_f>t_0$ such that all agents of that chain merges with other chains to form a chain in $\mathcal{FC}$ in $\C(t_f)$. Also, $t_f-t_0 \approx O(n)$ for all of these cases. 
\end{proof}

\begin{corollary}
    \label{Cor: Phi(t) becomes 0}
Let in $\C(t_0)$ all agents are in state \MergeI~ and $\Phi(t_0)>0$. Then $\exists~t_f>0$ such that in $\C(t_f)$, $\Phi(t_f)=0$ and agents are in state \In. 
\end{corollary}
\begin{proof}
    Since agents are in state \MergeI~ and $\Phi(t_0)>0$, by Lemma~\ref{lemma:merge1 decrease unwanted chains}, there exists a round, say $t_1$ within $O(n)$ rounds from $t_0$ such that $\Phi(t_1)\le \Phi(t_0)-1$. Now if $\Phi(t_1)= 0$, we are done. Otherwise, in $\C(t_1)$, there must exist at least one chain in $\mathcal{NFC}$, and, so agents are still in state \MergeI. Now using this same argument recursively we can conclude that within $O(ln)$ ($l$ being the number of agents) rounds from $t_0$, there exists a round $t_f'$ such that $\Phi(t_f)$ becomes 0 and agents are still at state \MergeI~in $\C(t_f')$. Now, in $\C(t_f'+1)$, since there are no chains in $\mathcal{NFC}$, agents change state to \In. So, taking $T_f=T_f'+1$ is sufficient for the proof. 
\end{proof}
\begin{lemma}
\label{lemma: 0- merge-1 to FC}
    Let $\C(t_0)$ be a configuration in which all agents are in state \mix. Then, within finite rounds, there will be a time $t_f$ such that exactly one of the following occurs.
    \begin{enumerate}
        \item all agents achieve chirality at some round $t$ such that $t_0\le t\le t_f$.
        \item In $\C(t_f)$, $\Phi(t_f)=0$ and all agents are in state \In.
    \end{enumerate}
\end{lemma}
\begin{proof}
    In $\C(t_0)$, all agents are in state \mix. Now, if $\C(t_0)$ does not contain any singleton 0-chain or, has at least one chain in $\mathcal{FC}$, then all agents change their state first to \In~. Let at $\C(t_0+1)$, all chains be in $\mathcal{FC}$. For this case, taking $T_f=t_0+1$ suffices. On the other hand, if at $\C(t_0+1)$ there is at least one chain in $\mathcal{NFC}$, then within constant round the agents changes state to \MergeI and by Corollary~\ref{Cor: Phi(t) becomes 0}, there exists a time $t_f\le t_0+O(ln)$ such that in $\C(t_f)$, either chirality is achieved, or $\Phi(t_f)=0$ where all agents are in state \In, and we are done.  
    
    \noindent So, let us now assume that $\C(t_0)$ contains at least one  0-chain and no chains are in $\mathcal{FC}$. Also, as described in Lemma~\ref{lemma:merge1 decrease unwanted chains}, an agent changes state to \mix~ at the third round of execution of Procedure~\ref{algo:preProcess}. So, all chains of $\C(t_0)$ that are in $\mathcal{NFC}$ must be singleton chains. At this point, as the agents are in state \mix,~ they begin executing Procedure~\ref{algo:merge-0to1} (Procedure \textsc{Merge-0To1()}). Note that, in $\C(t_0)$ there must exists two adjacent chains one singleton 0-chain $Ch_1 = C_{\D}(u_1,u_1)$ and another singleton 1-chain $Ch_2=C_{\D}(u_2,v_2)$. According to the procedure~\ref{algo:merge-0to1}, agent in $Ch_1$, say $r$, moves in direction $\D$, towards $Ch_2$ until it reaches $u_2'$ at time say $t_1$ for some $t_1>t_0$, where $u_2'$ is adjacent to $u_2$ in direction $\D$. Note that between $t_0$ and $t_1$, if the adversary tries to stop any move by some agent, the configuration becomes Asymmetric and agents agree on a chirality, and then we are done. So, let us assume chirality can never be achieved in finite time. In this case, within finite time (i.e., in $O(n)$ rounds), the agent $r$ merges with $Ch_2$ from the configuration $\C(t_0)$ and together they form a new chain $C_{\D}(u_2',v_2) \in \mathcal{FC}$. Now, as argued initially in this proof, within $O(1)$ rounds after $t_1$, there exists a time $t_1'$ when the agents change state to either \MergeI, or \MergeII, or \Osc~.
    In $\C(t_1')$ agents were in state \In~ with $\Phi(t_1')=0$ if agents change state to \MergeII, or, \Osc, at round $t_1'$. So, taking $t_1'=t_f$ is sufficient for this case. Now in $\C(t_1')$ if agents change state to \MergeI, then from Corollary~\ref{Cor: Phi(t) becomes 0} ,
     within further $O(ln)$ rounds there exists a round $t_2$, such that either agents achieve chirality at round $t_2$ or,  agents will have $\Phi(t_2)=0$ and state \In~in $\C(t_2)$. So, taking $t_f=t_2$ suffices for this case. 
     
\end{proof}

\begin{lemma}
\label{lemma:preround to certain scenarios}
    Let $\C(t_0)$ be the first configuration where all agents are in state \preRound\ for some $t_0>0$. Then, exactly one of the following statements is true.
    \begin{enumerate}
    \item there is a $t_1<t_0+2n+B+1$ such that in $\C(t_1)$ all agents change state to \In, then in further constant rounds to one of the following states \mix~or \MergeI, or \Osc
        \item  in $\C(t_0+2n+B+1)$, all agents achieve chirality.
        \item in $\C(t_0+2n+B+1)$, there exists a chain from the class $\mathcal{FC}$.
    \end{enumerate}
\end{lemma}
\begin{proof}
    Since the agents change state to \preRound, in $\C(t_0)$, all chains are singleton 1-chains. Then the agents start to execute Procedure~\ref{algo:roundTheRing-1} for $2n+B+2$ rounds. We first claim the following.
    \begin{claim}
        Let all agents at $\C(t_0)$ be in state \preRound~ for the first time. Then, exactly one of the following is true.
        \begin{enumerate}
            \item $\exists~ t_1\le t_0+B $ such that in $\C(t_1)$ all agents changes state to \In~ and then in further $O(1)$ rounds change state to exactly one of  \mix, \MergeI~ or, \Osc.
            \item in $\C(t_0+B+1)$ all chains are singleton 1-chain, and for any chain, both agents have the same chirality.
        \end{enumerate}
    \end{claim}
   
   \textbf{Proof of Claim:}
    Note that during the first $B$ execution of Procedure~\ref{algo:roundTheRing-1}, in round $t_0+x$ ($0<x<B-1$), each agent checks the $x+1$-th bit of their ID and if it is 1 it moves in the direction of the other agent in the chain otherwise it stays at the same node. Note that, if at some round $t_0+x$ both have their $x+1$-th bit 1, they cross each other and the 1-chain remains a 1-chain. This is also true if both of them have their $x+1$th bit as 0. Note that, during this move, the adversary can interrupt the movements by removing one edge of a chain making the configuration symmetric. In this case, the agents, say $r_1$ and $r_2$, in the chain with a missing edge move outwards on that round, and all agents change state to \In. 
    Note that since $l\ge 3$ in $\C(t_0)$, there are at least two chains.
    Now, suppose that after this move, agents $r_1$ and $r_2$ each create a singleton 0-chain. In this case, the agents currently in state \In, eventually transition to state \mix~ in 3 more rounds.
On the other hand, if either $r_1$ or $r_2$ merges with another chain during the outward move and forms an $i$-chain with $ i\ge 2$, then the agents transition to state \MergeI~ within 3 more rounds, provided that there are still chains remaining in the set  $\mathcal{NFC}$.
If no such chains remain, the agent eventually transitions to the state \Osc in the next round. So, if we assume that, to interrupt the agent's move, the adversary removes an edge to make the configuration symmetric, then this leads to the first point of claim 1. 

So, let us now assume that the adversary never removes any edge to interrupt the agent's move that would make the configuration symmetric. For this case, let us choose any singleton 1-chain arbitrarily. Let $r_1$ and $r_2$ be the agents on this chain. Let the ID of $r_1$ and $r_2$ first differs at the $y$-th bit from the right. So at round $t_0+y-1$, one of $r_1$ and $r_2$ moves towards the other, and the other remains. So, they form a 0-chain with two agents. Now, since all agents have different IDs and all of them differ at some bit $y\le B$, at $\C(t_0+B)$ all chains are 0-chains with two agents. Now, let us consider one such arbitrary 0-chain with two agents $r_1$ and $r_2$. At round $t_0+B$, the lowest ID agent of them (W.L.O.G say $r_1$) moves in its clockwise direction, making each chain a singleton 1-chain again. Now, at round $t_0+B+1$ when $r_2$ sees that $r_1$ is one hop away in a particular direction, $r_2$ also agrees on the same clockwise direction as $r_1$. We call a singleton 1-chain \textit{locally chiral}, if both the agents on the chain have the same clockwise direction. So, we proved that, if the adversary does not interrupt any move, all chains will be locally chiaral in $\C(t_0+B+1)$. So, we prove that point 2 of the claim is true if point 1 is not true.
\qed \\
So, if not already changed to one of the following states i.e., \mix, \MergeI, \Osc, from round $t_0+B+1$, all agents start to move in clockwise direction for $2n$ consecutive rounds if they are still part of 1-chains and no missing edge in their chain. We claim that if all agents do not have the same clockwise direction, then in this $2n$ round, there must exist one pair of chains that get merged and form a chain in $\mathcal{FC}$.  Note that, since $l\ge3$, there must be at least two chains in $\C(t_0)$, which assures the existence of at least two singleton 1-chains. Let us assume all agents do not have the same direction agreement in $\C(t_0+B+1)$. This implies there are two locally directed adjacent singleton 1-chains, say $Ch_1$ and $Ch_2$, whose clockwise direction faces each other. Now, the adversary can stop them from merging if both of these chains are forced not to move for at least $n$ rounds each, as the distance between them can be at most $n$. But, if $Ch_1$ is stopped for $n$ rounds, $Ch_2$ moves for those $n$ rounds, and vice versa. So, $Ch_1$ and $Ch_2$ merge and form one chain in $\mathcal{FC}$. So, it is evident that, if all agents do not have the same clockwise direction in $\C(t_0+B+1)$, then in $\C(t_0+2n+B+1)$ at least one chain from the class $\mathcal{FC}$ exists. Otherwise, if all agents have the same direction in $\C(t_0+B+1)$ then the adversary can create both the following configurations as $\C(t_0+2n+B+1)$. Either, all chains are still singleton 1-chain or, there exists at least one chain from $\mathcal{FC}$. So, if in $\C(t_0+2n+B+1)$ all chains are still singleton 1-chain, we can conclude that agents have achieved chirality or, in $\C(t_0+2n+B+1)$, there must exists a chain in $\mathcal{FC}$ and agents change state to \In~ in the next round (According to Procedure~\ref{algo:roundTheRing-1}). This concludes the proof.
\end{proof}

\begin{corollary}\label{cor: preround to phi 0}
    Let $\C(t_0)$ be the first configuration where all agents are in state \preRound. Then in $\C(t_0)$, $\Phi(t_0)=\frac{l}{2}$ and  within $O(ln)$ rounds, there exists a round $t_f$ such that exactly one of the following is true.
    \begin{enumerate}
        \item all agents achieves chirality in $\C(t_f)$
        \item In $\C(t_f)$, $\Phi(t_f)=0$ and all agents are in state \In.
    \end{enumerate}
\end{corollary}
\begin{proof}
    Let $\C(t_0)$ be the first configuration when all agents are in state \preRound. This implies that, in $\C(t_0)$, all chains are singleton 1-chains. Thus, each chain is from the class $\mathcal{NFC}$  and contains exactly two agents. So $\Phi(t_0)=\frac{l}{2}$. Now, from Lemma~\ref{lemma:preround to certain scenarios}, one of the following three scenarios can occur. 
    \begin{enumerate}
    \item there is a $t_1<t_0+2n+B+1$ such that in $\C(t_1)$ all agents change state to \In~ and in further constant rounds to one of the following states \mix~or \MergeI, or \Osc
        \item  in $\C(t_0+2n+B+1)$, all agents achieve chirality.
        \item in $\C(t_0+2n+B+1)$, there exists a chain from the class $\mathcal{FC}$ and agents change state to \In.
    \end{enumerate}

    Now, for case 2, we are already done. For Case 1, at $\C(t_1)$ agents are in state \In, $t_1<t_0+2n+B+1$. Now, if in $\C(t_1)$ there are no chains in $\mathcal{NFC}$. in this case agents change state to \Osc in round $t_1$ from \In. So, taking $t_f=t_1$, we are done for this case. So let there be chains from class $\mathcal{NFC}$ in $\C(t_1)$. So, within $O(1)$ further rounds, there is a round $t_2$ such that at $t_2$, the agents change state to either \mix~or \MergeI.
    Now if at $\C(t_2)$ agents are at state \MergeI, then by Corollary~\ref{Cor: Phi(t) becomes 0} within $O(ln)$ further rounds there exists a round $t_f$ such that either chirality is achieved in $\C(t_f)$, or $\Phi(t_f)=0$ with each agents being in state \In. Similarly, we can say the same for the case where in $\C(t_2)$ all agents change to state \mix~ from Lemma~\ref{lemma: 0- merge-1 to FC}. 
    
    For Case 3, at $\C(t_0+2n+B+1)$, there exists a chain from the class $\mathcal{FC}$. This implies that, in the next configurations all agents will have state \In (according to Procedure~\ref{algo:roundTheRing-1}) and in constant further rounds there is a round $t_3$ such that in $\C(t_3)$ agents either change to state \MergeI, or \Osc~or \MergeII. 
    
    Now, if agents change state to either \Osc, or \MergeII~ at $t_3$ then in $\C(t_3)$, $\Phi(t_3)=0$ and all agents are in state \In. So, taking $t_f=t_3$ is sufficient for this case.

    Now for the case where in $t_3$ agents change state to \MergeI, by Corollary~\ref{Cor: Phi(t) becomes 0}, within further $O(ln)$ rounds there exists a round $t_4$ such that either agents achieve chirality at $t_4$ or, in $\C(t_4)$, $\Phi(t_4)=0$ and all agents are in state \In. Now, taking $t_f=t_4$ suffices for this case.
 \end{proof}

 \begin{lemma}\label{lemma: round to certain cases}
     Let $\C(t_0)$ be the first configuration when all agents are in state \round, for some $t_0>0$. Then, exactly one of the following statements is true.
     \begin{enumerate}
         \item $\exists~ t_1 \le t_0+n$ such that at $\C(t_1)$ all agents achieve chirality.
         \item $ \exists ~t_1 \approx t_0+n+o(1)$ such that at $\C(t_1)$, all agents change state to exactly one of the following.
         \preRound, \MergeI, \Osc, or, \MergeII
     \end{enumerate}
 \end{lemma}
 \begin{proof}
     Let all agents at $\C(t_0)$ be in state \round~ for the first time. Then, all chains in $\C(t_0)$ are singleton 0-chains. Thus $\Phi(t_0)=l$. In this case, the agents start executing Procedure \ref{algo:roundTheRing-0} for $n$ consecutive rounds. In this procedure, an agent moves along its clockwise direction for $n$ consecutive rounds unless within these $n$ rounds it forms either a 0-chain with multiplicity, or an $i-$chain, where $i\ge 1$. Note that, during this move, the adversary can not stop the movement of an agent by removing an edge and without making the configuration asymmetric. And, agents can achieve chirality from an asymmetric configuration So, if the adversary stops movement by any agent at a round $t_1< t_0+n$, agents can achieve chirality in configuration $\C(t_1)$. Now, let the adversary not try to interrupt any agent by removing any edge. In this case, if all agents in $\C(t_0)$ have the same clockwise direction. Then, at $\C(t_0+n)$, all agents will remain part of a singleton 0-chain. So, if this is the case, then at $\C(t_0+n)$ agents can agree upon the chirality. Now, if at $\C(t_0)$, agents do not have the same clockwise direction, then there exist two adjacent agents and they move towards each other. Now, they either form a 0-chain with multiplicity or one 1-chain. Now, since $l\ge 3$, other agents may also merge in this chain and form a 0-chain with multiplicity or,  visibly directed  1-chain, or, visibly undirected 1-chain, or, $i-$chain where $i\ge 2$. Note that, no singleton 0-chain would remain, as it would imply that a singleton 0-chain from $\C(t_0)$ completed $n$ clockwise moves without meeting any other, which is not possible as agents have different clockwise directions. For this case, at $\C(t_0+n)$ agents change state to \In.

     Now, in the next round (i.e., $t_0+n+1$), if all chains are visibly directed, then either the configuration is global (in that case chirality is achieved), or they change state to \MergeII. 

     Now, if all chains in round $t_0+n+1$ are not visible directed, but all chains are from the class $\mathcal{FC}$, then either the configuration is global for which all agents achieve chirality or, they change state to \Osc.  For the above two cases, taking $t_1=t_0+n+1$ is sufficient for the proof.

     On the other hand, if at $t_0+n+1$, all chains are not from class $\mathcal{FC}$, then either they achieve chirality (if the configuration is global or asymmetric) or, they change to state \PreMergeI. If agents change state to  \PreMergeI, then agents execute Procedure~\ref{algo:preProcess} for 3 consecutive rounds. In the first two rounds, all 0-chains with multiplicity become either a visibly directed 1-chain or, singleton 1-chain, and all 1-chains with multiplicity on both vertices become $i-$chain where $i\ge 2$. So, at the third round of executing Procedure~\ref{algo:preProcess} (i.e., at round $t_0+n+4$), if all chains are visibly directed then either chirality is achieved, or agents change state to \MergeII. On the other hand, if all chains are not visibly directed, but all are from $\mathcal{FC}$, then agents either achieve chirality (if the configuration is global), or they change state to \Osc.
     Now there is further one case that can occur at round $t_0+n+4$, that is at this round, there are chains from the class $\mathcal{NFC}$, in this case, either agents achieve chirality (if configuration is global or, asymmetric)or,  agents changes to states either \preRound~ or, \MergeI. (state \mix~ does not occur as there can not be any 0-chain at $t_0+n+1$ as described above). So, for the above cases, taking $t_1=t_0+n+4$ is sufficient. 
 \end{proof}

 \begin{corollary}\label{cor: roundtozero}
     Let $\C(t_0)$ be the first configuration when all agents are in state \round, for some $t_0>0$. Then,
   $\Phi(t_0)=l$ and there exists a round $t_f\le t_0+O(ln)$ such that in $\C(t_f)$ exactly one of the following is true
     \begin{enumerate}
         \item all agents achieve chirality.
         \item  $\Phi(t_f)=0$ and all agents are in state \In. 
     \end{enumerate}
 \end{corollary}
\begin{proof}
    From Lemma~\ref{lemma: round to certain cases} it is clear that at $\C(t_0)$, $\Phi(t_0)=l$. Also from the same lemma we can conclude that $\exists~ t_1 < t_0+O(n)$,  such that at $t_1$ either chirality is achieved or, at $t_1$ agents are in one of the following states \preRound, \MergeI, \MergeII, or, \Osc.

    If at $t_1$ agents change state to \Osc, or \MergeII then we take $t_f=t_1$. In this case, $\Phi(t_f)=0$. And, in $\C(t_f)$, all agents are in state \In. On the other hand, if at $t_1$ all agents are in state either \preRound, or \MergeI, then by Corollary~\ref{Cor: Phi(t) becomes 0}, and Corollary~\ref{cor: preround to phi 0}, within further $O(ln)$ rounds there will be a round $t_2$ such that at $t_2$ either agents achieve chirality or, in $\C(t_2)$, $\Phi(t_2)=0$ with all agents in state \In. So, for this case we take $t_f=t_2$, which proves the result.
\end{proof}

\begin{theorem}\label{Thm:Phi0}
For any 1-interval connected ring with $l\ge 3$ agents executing \textsc{Achiral-2-Chiral}, there exists a round $t_f< O(ln) $ such that in $\C(t_f)$, either chirality is achieved, or, $\Phi(t_f)=0$ where all agents are in state \In.
\end{theorem}
 \begin{proof}
     Let $\C(0)$ be any initial configuration with number of agents $l\ge 3$. Now in $\C(0)$ if all chains are in $\mathcal{FC}$ then $t_f=0$, and we are done. So, let there exist chains from the class $\mathcal{NFC}$ in $\C(0)$. So, when agents execute Algorithm~\ref{algo:Main}, they change state to state \PreMergeI. In state \PreMergeI, agents execute Procedure~\ref{algo:preProcess}. The agents execute this process for three consecutive rounds. In the first two rounds, if the configuration is not global or asymmetric, then all 0-chains with multiplicity become either a visibly directed 1-chain or a singleton 1-chain. And all 1-chain, with multiplicity on both vertices, becomes $i-$chain, where $i\ge 2$. Otherwise, chirality will be achieved, and we are done. So, let the configuration not be global or asymmetric in the first two rounds of executing Procedure~\ref{algo:preProcess}. Thus, at the beginning of the third round of executing Procedure~\ref{algo:preProcess}, say $t_1$, there remains no 0-chain with multiplicity or, 1-chain that has multiplicity on both the vertices. So, in $\C(t_1)$ the chains can be either a singleton 0-chain, a singleton 1-chain, or any other chains from the class $\mathcal{FC}$. Now $\C(t_1)$ can be exactly one of the following configurations,
     \begin{itemize}
         \item All chains are singleton 0-chain in $\C(t_1)$. For this case, agents change state to state \round.
         \item All chains are singleton 1-chain in $\C(t_1)$. For this state, agents change state to \preRound.
         \item $\exists$ a singleton 0-chain and a singleton 1-chain and $\nexists$ any chain from $\mathcal{FC}$. For this case, agents change state to \mix.
         \item $\exists$ a chain from $\mathcal{FC}$. For this case, agents change state to \MergeI.
     \end{itemize}
     So, at $t_1+1$, agents can be at any one of the following states, \round, \preRound, \mix, \MergeI. 
     Now, by Corollary~\ref{Cor: Phi(t) becomes 0}, Lemma~\ref{lemma: 0- merge-1 to FC}, Corollary~\ref{cor: preround to phi 0} and Corollary~\ref{cor: roundtozero}, Within further $O(ln)$ rounds from $t_1$,  there will be a round $t_f$ such that at $t_f$ chirality is achieved by all the agents or, $\Phi(t_f)=0$.
  \end{proof}

 \begin{theorem}
     \label{thm: all visibly directed}
     Let $\C(t_0)$ be a configuration for some $t_0\ge0$ such that in $\C(t_0)$, all agents are in state \In~ and $\Phi(t_0)=0$. Let $\C(t_0)$ has $p$ chains such that not all of them are visibly directed. Then, $\exists~ t_f \ge  t_0$ such that in $\C(t_f)$ the agents either achieve chirality or, in $\C(t_f)$, all of the $p$ chains become visibly directed and in state \In. 
 \end{theorem}
 \begin{proof}
     Let $\C(t_0)$ be a configuration where all agents are in state \In~ and $\Phi(t_0)=0$. The existence of such a round is ensured by Theorem~\ref{Thm:Phi0}. Now if the configuration is global or, asymmetric agents achieve chirality at $t_0$. So, taking $t_f=t_0$ would suffice. Now let $\C(t_0)$ be neither global nor asymmetric. Now, since there are chains, not visibly directed, the agents change state to \Osc~ at round $t_0$. Now, from round $t_0+1$ to round $t_0+2B+2$, all agents execute Procedure~\ref{algo:GetDirected}. Let $C_D(u,v)$ be a visibly undirected chain in $\C(t_0+1)$ and at least one of $u$ and $v$ is of multiplicity. Note that $C_D(u,v) \in \mathcal{FC}$. So, by definition, $d_{\D}(u,v) \ge 2$. Now, in the first execution of procedure~\ref{algo:GetDirected}, only the lowest ID agents stay at $u$ and $v$. All other agents of $u$ and $v$ (if they exist) move inside the chain, so that both $u$ and $v$ become singletons. Let $r_u$ be the agent at $u$ and $r_v$ be the agent at $v$ in $\C(t_0+2)$. Now for the next $2B$ rounds starting from $t_0+2$, for round $t_0+2x+2~ (0\le x< B)$, both $r_u$ and $r_v$ check their $(x+1)$-th bit from right and if the bit is 1, they move inwards the chain. Otherwise, they do not move. Now, for each round $t_0+2x+3 ~(0\le x< B)$, if any agent among $r_u$ and $r_v$ moved in round $t_0+2x+2$, it moves back. Note that, since the chains are $i$-chains with $i\ge 2$, halting the agents at nodes $u$ and $v$, to move inward and then outward, effectively removing an edge, results in an asymmetric configuration, which ensures that the agents achieve chirality. So, let us assume the agents do not interrupt $r_u$ and $ r_v$'s movement. Also note that the movement by $r_u$ and $r_v$ ensures that no chain breaks into multiple chains or no two or more chains merge during this procedure. So the total number of chains remains the same in each configuration from round $t_0$ to $t_0+2B+2$. In that case, since $ID(r_u)\ne ID(r_v)$, there exists $y\le B$ such that the $y$-th bit from the right is the first bit that is different for $r_u$ and $r_v$. So, at round $t_0+2y$, only one of $r_u$ and $r_v$ moves inward the chain. So, in $\C(t_0+2y+1)$ the chain $C_{\D}(u,v)$ becomes visibly directed. From $t_0+2y+1$ to $t_0+2B+1$, agents in the visibly directed chain do not do anything. The above argument is true for any visibly undirected chain in $\C(t_0+1)$. So at $\C(t_0+2B+2)$ all chains become visibly directed. Further, all agents change state to \In. This concludes the proof.
 \end{proof}

 \begin{theorem}
     \label{thm: chain decrease}
Let $\C(t_0)$ be a configuration where all agents are in the state \In~ and all chains are visibly directed. If the configation $\C(t_0)$ contains $p >1$ such chains, then within $O(n)$ rounds, there exists a round $t_f$, at which the configuration $\C(t_f)$ satisfies one of the following:  
\begin{itemize}
    \item all agents agree on a common chirality, or
    \item all agents remain in the state \In, all chains are visibly directed, and the number of chains has reduced to $p' < p$.
\end{itemize}
\end{theorem}  
\begin{proof}
     Let $\C(t_0)$ be a configuration in which all agents are in state \In~ and all $p$ chains are visibly directed. If the configuration is global or asymmetric, then we are done, as chirality will be achieved by the agents. So, let the configuration $\C(t_0)$ be neither global nor asymmetric. So all agents change their state to \MergeII. This implies that among $p$ chains there exists at least a pair of adjacent chains whose direction is towards each other.  Let us arbitrarily choose one such pair, and name them $Ch_1=C_{\D}(u_1,v_1)$ and $Ch_2= C_{\D}(v_2,u_2)$. During the state \MergeII, agents execute Procedure~\ref{algo:Merge-II}. In this procedure, the chains $Ch_1$ and $Ch_2$ move towards each other by moving according to the direction of the chains if there is no missing edge in of $Ch_1$ or $Ch_2$. Without loss of generality, let $Ch_1$ has a missing edge. If the configuration is asymmetric due to the missing edge, then the agents achieve chirality, and we are done. On the other hand, if the configuration is symmetric, then agents in $Ch_1$ do not move, but agents in $Ch_2$ move towards $Ch_1$. So, the distance between $Ch_1$ and $Ch_2$ decreases at least by one and at most by two in each execution of Procedure~\ref{algo:Merge-II}. So, there exists a time $t_1$  such that both $Ch_1$ and $Ch_2$ merges and creates a new chain $Ch_3 \in \mathcal{FC}$ in the configuration $\C(t_1)$ in expense of two chains(as the chain length will be greater or equals to 2). Note that, $Ch_3$ will not be visibly directed as both the terminal nodes are singletons also as the initial distance between $Ch_1$ and $Ch_2$ can be at most $O(n)$, to $t_1-t_0\approx O(n)$. This will be true for all such pairs of visibly directed chains directed towards each other in $\C(t_0)$. Let there be $q$ such pairs. So, until all $q$ pairs are merged, the agents in the already merged chains do nothing. Now, there will be a time $t_2\ge t_1$ such that at $t_2$ there are $p'=p-q$ chains where $q>0$ and all chains are in $\mathcal{FC}$ but not all chains are visibly directed. Note that, for the same reason stated above, $t_2-t_0\approx O(n)$. In $\C(t_2)$, all agents change state from \MergeII~ to \In. Then at $t_2+1$, all agents are in state \In and $\Phi(t_2+1)=0$. Then by Theorem~\ref{thm: all visibly directed}, within $O(B)=O(\log n)$ rounds from $t_2$ there exists a round $t_3\ge t_2$ such that either agents achieve chirality in $\C(t_3)$ or, in $\C(t_3)$ all $p'=p-q<0$ chains becomes visibly directed along with all agents being in state \In. For this case, taking $t_f=t_3$ proves the result.
\end{proof}
\begin{theorem}
In any $1$-interval connected ring with $n$ nodes and $l$ achiral agents ($n \ge l \ge 3$), where each agent has $O(\log n)$ memory and the agents are initially placed arbitrarily on the nodes, chirality can be achieved in $O(ln)$ rounds by executing the \textsc{Achiral-2-Chiral} algorithm.
\end{theorem}
\begin{proof}
    By Theorem~\ref{Thm:Phi0}, we can conclude that for any arbitrary initial configuration $\C(0)$, there exists a round $t_1$ within $O(ln)$ rounds from the initial configuration, such that if chirality is not achieve till $t_1$ then in $\C(t_1)$ all agents are in state \In~ and $\Phi(t_1)=0$. Now, in $\C(t_1)$ if all chains are not visibly directed then by Theorem~\ref{thm: all visibly directed} we can conclude that within $O(B)=O(\log n)$ more rounds there exists a round $t_2$ such that if chirality is not achieved till $t_2$,  all chains become visibly directed chains with all agents being in state \In, in $\C(t_2)$. Now, let at $\C(t_2)$, there are $p$ chains. If $p=1$, then the direction of this chain is the agreed direction for all the agents. On the other hand, let $p>1$. For this case, by Theorem~\ref{thm: chain decrease}, within further $O(n)$ rounds, there exists a round $t_3$, such that if chirality is not achieved by $t_3$, in $\C(t_3)$, all agents are in state \In, all chains visibly directed and number of chains, say $p' <p$. If, $p'\ne 1$, and chirality is not achieved, then then we can recursively use Theorem~\ref{thm: chain decrease} until some time $t_4>t_3$ such that if until $t_4$ chirality is still not achieved then at $\C(t_4)$, all agents are in state \In~ there is exactly one visibly directed chain. In this case, agents can agree on chirality using the direction of the chain.
    Note that in the worst case, there can be $O(l)$ chains in $t_2$ and each $O(n)$ further rounds, the number of chains reduces by one. This way, we can conclude that $t_4-t_2\approx O(ln)$. So, if we take $t_f=t_4$, we can conclude that, from any initial configuration within $O(ln)$ rounds, the agents will agree on a chirality. 
\end{proof}

\subsection{Algorithm \mdseries{\textsc{Dispersed}}}
In this subsection, we describe the Algorithm \textsc{Dispersed}. After chirality is achieved by the agents and if the configuration is not dispersed, the agents first execute Algorithm \textsc{Dispersed} to achieve dispersion before achieving distance-$k$-dispersion. Note that,  achieving dispersion partially solves the open problem regarding dispersion from any arbitrary configuration for any 1-interval connected ring proposed in \cite{AAMKS2018ICDCN}. Let us have a high-level idea of the algorithm.

\subsubsection{Overview of \mdseries{\textsc{Dispered}}}
Now, we provide a brief overview of the subroutine \textsc{Dispersed()}, which is designed to transform the multiplicity nodes into singleton nodes to achieve dispersion. Let a configuration have a multiplicity node in some chain $C_{CW}(T, H)$, where $T$ and $H$ are the first node and last node, respectively, called $tail$ and $head$, of the chain in clockwise direction. First, we try to convert the multiplicity node, $M$ say, that is nearest to the tail of the chain.

\begin{algorithm}[H]
\caption{\textsc{Dispersed}($r$,$\mathcal{C}(t_c)$)}
\label{Dispersed}
$C_{CW}(T,H) = $ chain in which $r$ is located\;
 
\LinesNumbered
\If{$C_{CW}(T, H)$ contains a multiplicity}
    {
    $M=$ multiplicity node in $C_{CW}(T,H)$ nearest to $T$.\;
    $T'=$ Node adjacent to $T$ in $CCW$ direction\;
    $H'=$ Node adjacent to $H$ in $CW$ direction\;
    
      \uIf{$e_t \notin (T',H')_{CW} $ $\lor$ $e_t \in$ $(T',M)_{CW}$}
         {
          \uIf{$r$ is at $M$}
               {
                \If{$r$ is not least ID agent at $M$}{move in clockwise direction.}
               }
          \ElseIf{$r \in C_{CW}(M,H) \land r \notin M$}
               {move in clockwise direction.}
              
         }    
     \ElseIf{$e_t \in (M,H')_{CW}$}
         {
          \uIf{$r$ is at $M$}
               {
                \If{$r$ is not theleast ID agent at $M$}{move in counterclockwise direction.}
               }
          \ElseIf{$r \in C_{CW}(T,M) \land r \notin M$}
               {move in counterclockwise direction.}
         }
    
    }
\end{algorithm}

Let the chain be partitioned into two sub-chains, $C_{CW}(T, M)$ and $C_{CW}(M, H)$ in clockwise direction with respect to the multiplicity $M$. Then, in this scenario, three cases may occur according to the position of the missing edge $e_t$ (if it exists).
\begin{itemize}
    \item []\textbf{Case-I:} There is no missing edge in the arc $(T',H')_{CW}$, where $T'$ is the adjacent node of $T$ in counterclockwise direction and $H'$ is the the adjacent node of $H$ in clockwise direction.
    \item []\textbf{Case-II:} The missing edge is in the arc $(T',M)_{CW}$.  
    \item []\textbf{Case-III:} The missing edge is in the arc $(M,H')_{CW}$. 
\end{itemize}
For, first and second case, i.e., when there is there is no missing edge in arc $(T', H')_{CW}$ or, the missing edge is in the chain $(T', M)_{CW}$, then the agent with least ID on the multiplicity $M$ stay at its node and the other agents of the sub-chain $C_{CW}(M, H)$, except the least ID agent at $M$, move in clockwise direction. This way, $M$ becomes a singleton node.
Now, for the third case, i.e., when the missing edge is in the arc $(M, H')_{CW}$, the agent with least ID on the multiplicity $M$ stay at its node and the other agents of the sub-chain $C_{CW}(T, M)$, except the least ID agent at the multiplicity $M$, move in counterclockwise direction.
In each round during the execution of algorithm \textsc{Dispersed}, the number of agents on a multiplicity node decreases. Now, since there can be at most $O(l)$ multiplicity points, within $O(l)$ rounds, the configuration becomes a dispersed configuration.

\subsubsection{Correctness and Analysis of Algorithm \mdseries{\textsc{Dispersed}}}
\label{App C}
Let us first discuss some necessary notation for the correctness proof. We define the number of unoccupied nodes in a configuration $\C(t)$ with $\psi(t)$ for some time $t$. Note that if at some round $t$, the configuration $\C(t)$ is dispersed, then $\psi(t)=n-l$. Otherwise $\psi(t)>n-l$. 
To prove the correctness of the algorithm \textsc{Dispersed}, we prove the following lemma first.
\begin{lemma}
\label{lemma: hole node fills}
    Let $\C(t)$ be a configuration where $\psi(t)>n-l$. Then $\psi(t+1)<\psi(t)$ if agents execute \textsc{Disperse} in round $t$.
\end{lemma}
\begin{proof}
    Since $\psi(t)>n-l$, there must exist a chain with multiplicity. Let $Ch_1=C_{CW}(T, H)$ be any such chain. Let $T'$ and $H'$ be the nodes adjacent to $T$ and $H$in the direction $CCW$ and $CW$ respectively. Note that $T'$ and $H'$ are both unoccupied nodes in $\C(t)$. Let $M$ be the multiplicity node in $Ch_1$ that is nearest to $T$. Then, there can be three cases depending on the position of the missing edge $e_t$ (if it exists). These three cases are,
    \begin{enumerate}
        \item $e_t \notin (T',H')_{CW}$
        \item $e_t \in (T',M)_{CW}$
        \item $e_t \in (M,H')_{CW}$
    \end{enumerate}
    For cases 1 and 2, all agents of $Ch_1$ in the arc $(M, H)_{CW}$ except the lowest ID agent at $M$ move clockwise according to algorithm \textsc{Dispersed}. Note that due to this $H'$ along with becomes occupied in $\C(t+1)$. Also, while executing algorithm \textsc{Dispersed}, no occupied nodes become unoccupied. So for these cases $\psi(t+1)<\psi(t)$. 
    Now for case 3, all agents in arc $(T, M)_{CW}$ move counterclockwise except the lowest ID agent at $M$ according to algorithm \textsc{Dispersed}. This ensures that $T'$ becomes occupied in $\C(t+1)$. Also, as no occupied nodes in $\C(t)$ become unoccupied in $\C(t+1)$, we can say $\psi(t+1) < \psi(t)$ in this case too.
\end{proof}
Now, using the above lemma, we prove the following theorem that justifies the correctness of algorithm \textsc{Dispersed}. 
\begin{theorem}
    \label{thm:dispersed}
   If at round $t$, the configuration $\C(t)$ is not dispersed but agents have chirality agreement, then a dispersed configuration $\C(t_f)$ will be reached within $O(l)$ rounds from $t$. 
\end{theorem}
\begin{proof}
    Since the configuration $\C(t)$ is not dispersed, there must exist a multiplicity node. Thus $\psi(t)>n-l$. Also, the agents have chirality agreement. Thus, at $t$ agents execute \textsc{Dispersed}. Now by Lemma~\ref{lemma: hole node fills},in $\C(t+1)$, $\psi(t+1)$ is at least one less than $\psi(t)$ (i.e., $\psi(t+1) \le \psi(t)-1$). Now, if at $\C(t+1)$, $\psi(t+1)$ is still strictly greater than $n-l$, using the same argument, we can say there exists a round $t_f$ such that $\psi(t_f)=n-l$. So, $\C(t_f)$ becomes a dispersed configuration. Note that in $\C(t)$ there can be at most $n-1$ unoccupied nodes. So $t_f-t$ is asymptotically equal to $O(l)$.
\end{proof}


\subsection{Algorithm \mdseries{\textsc{Dispersed To $k$-Dispersed}}}
When all agents achieve dispersion and if $k>1$, then agents execute the algorithm \textsc{Dispersed To $k-$Dispersed} to achieve Distance-$k$- Dispersion. Following this, a brief overview of the algorithm is discussed.

\subsubsection{Overview of the Algorithm}
The agents execute the algorithm \textsc{Dispersed-To-$k-$Dispersed} until the configuration is a distance-$k-$dispersed configuration. If the configuration is not distance-$k$-dispersed, there must exist a $k$-Link with more than one agent. In each round, a selected set of agents moves either clockwise or counterclockwise so that the distance between at least one pair of agents placed in adjacent occupied nodes in the same $k$-Link increases and no pair of agents placed at adjacent occupied nodes in the same $k$-Link decreases. Further, it also ensures that agents placed at adjacent occupied nodes, which are already at a distance more than or equal to $k$, never decrease beyond $k$.

\begin{algorithm}
\caption{\textsc{Dispersed-To-$k$-Dispersed}($r$,$\mathcal{C}(t_c)$)}
\label{Algo:D-k-D}
\eIf{$\C(t_c)$ is not distance-$k$-dispersed}
    {
     \eIf{$e_{t_c}$ is not incident to $v(r')$ in the clockwise direction for some $r' \in \mathcal{EAS}(t_c)$}
     {
        \If{$r \in \mathcal{EAS}(t_c)$}
        {
            move clockwise\;
        }
     }
     {  $r'=$ the agent in $\mathcal{EAS}(t_c)$ such that $e_{t_c}$ is incident to $v(r')$in the clockwise direction\;
        $KS =$ the $k$-Link such that $r' \in \mathcal{NS}_{CW}(KS)$ \;
        $KS'=$ the first movable $k$-Link in the clockwise direction starting from $KS$.
        \If{$r \in \mathcal{NS}_{CCW}(KS')$}
        {
            move counter clockwise\;
        }
     }
    }
    {
        terminate\;
    }
\LinesNumbered
\end{algorithm}

The priority is given to clockwise movement first in any particular round, say $t$. For clockwise movement, the selected set is $\mathcal{EAS}(t)$. Agents in $\mathcal{EAS}(t)$ only move clockwise in round $t$, if the missing edge $e_t$ is not adjacent to $v(r)$  for any $r\in \mathcal{EAS}(t)$ in the clockwise direction. Otherwise, let $r' \in \mathcal{EAS}(t)$ be the agent such that the incident edge of $v(r')$ in the clockwise direction is the missing edge. Let $r'$ be the part of the $k$-Link $KS$ and $KS'$ be the first movable $k$-Link in the clockwise direction starting from $KS$. Then the set of agents in the set of $\mathcal{NS}_{CCW}(KS')$ move counterclockwise. We now provide a brief overview of the correctness of the algorithm, ensuring that within a bounded time, any dispersed configuration becomes distance-$k$-dispersed.

\subsubsection{Correctness and Analysis of  \mdseries{\textsc{Dispersed-To-$k-$Dispersed}}}
\label{App D}
\begin{lemma}
    \label{lemma: movable link exists}
    If a configuration $\C(t)$ is not a distance-$k$-dispersed configuration for some round $t$, then there exists a movable $k$-Link in $\C(t)$.
\end{lemma}
\begin{proof}
    Let $KS_1, KS_2, \cdots, KS_{\alpha}$ be all the $k$-Links in $\C(t)$ and $|KS_i|=l_i$ for all $i \in \{1,2,\cdots \alpha\}$$=I$. Let  $I'\subseteq I$ such that $|KS_i|\ge 2, \forall i\in I'$ and $|KS_j|=1, \forall j\in I\setminus I'$. Let $Len(KS_i)= d_{CW}(T(KS_i), H(KS_i))$ for all $i \in I$. Then $Len(KS_i)<(l_i-1)k$ where $i\in I'$ and otherwise $Len(KS_j)=0$ (i.e., when $j\in I\setminus I'$).

    If possible, let all $KS_i$ be non-movable $k$-Links. Then the exact next $k$ nodes in the clockwise direction from the head of any $KS_i$ are unoccupied. Let, for the $k$-Link $KS_i$, the arc with the next $k$ unoccupied nodes is denoted as $U_i$. We denote $Len(U_i)=k$ for all $i \in I$. So for every $i \in I'$, $Len(KS_i)+Len(U_i)< l_ik$ and for every $j \in I\setminus I'$, $Len(KS_i)+Len(U_i)=k$.

    Now, $n=\sum_{i=1}^{\alpha} Len(KS_i)+Len(U_i)=\sum_{i\in I'}[Len(KS_i)+Len(U_i)]+\sum_{i \notin I'}[Len(KS_i)+Len(U_i)]<\sum_{i\in I'}(l_ik)+xk$ (as $I'$ is not an empty set and assuming $I\setminus I'$ has $x$ elements without loss of generality). Now, from the above, we can say $n<(l-x)k+xk=lk$. But this is a contradiction, as according to the assumption $n\ge lk$. Hence, among the $\alpha$ $k$-Links, at least one movable $k$-Link exists.
\end{proof}
\begin{lemma}
    \label{lemma: larger gaps cannot be lower than $k$} Let $r$ be an agent such that $r'$ is the next agent of $r$ in clockwise direction in $\C(t)$ for some round $t$. Let $d_{CW}(r,r')>k$ in $\C(t)$, then $\forall t_1>t$, $d_{CW}(r,r')\ge k$ in $\C(t_1)$.
\end{lemma}
\begin{proof}
    Let $\exists~ t'>t$ such that in $\C(t')$, $d_{cw}(r,r')<k$. This implies there exists $t''$ such that $t>t''>t'$ and in $\C(t'')$, $d_{CW}(r,r')$ decreases from $k$. This is possible only if, exactly one of $r$ or $r'$ stays still while the other moves in round $t''$. If $r$ moves in $\C(t'')$. Then, it must have moved clockwise to reduce the distance from $k$ further. This implies $r\in \mathcal{EAS}(t'') \implies r \in \mathcal{NS}_{CW}(KS)$, for some movable $k$-Link $KS$. This implies $r= H(KS)$. Thus $d(r,r')>k$  in $\C(t'')$(by definition of a movable $k$-Link) which leads to contradiction. 
    
    So, now let $r'$ move and $r$ stays still in $t''$. Then, in $\C(t'')$, $r'$ must have moved counterclockwise.  This implies $\exists~ k$-Link $KS_1$ such that $e_{t''}$ is incident  to $v(r_1)$ in the clockwise direction where $r_1=H(KS_1)$. Let $KS_1'$ be the first movable $k$-Link in the clockwise direction starting from $KS_1$. This implies, $r'\in \mathcal{NS}_{CCW}(KS_1') $ and $r\notin \mathcal{NS}_{CCW}(KS_1')$. In fact $r=H(KS_1')$ in $\C(t'')$. Thus $d_{CCW}(r,r')>k$ in $\C(t'')$ which is a contradiction to the assumption that $d_{CCW}(r,r')=k$ in $\C(t'')$. Thus, for all $t_1\>t$, $d_{CCW}(r,r')\ge k$ in $\C(t_1)$. 
\end{proof}

\begin{lemma}
\label{lemma: smaller gap never decrease}
    Let for some round $t$, $d_{CW}(r,r')=k'<k$ where $v(r)$ and $v(r')$ be two adjacent occupied nodes in $\C(t)$. Then $\forall t'>t$, $d_{CW}(r,r')\ge k'$ in $\C(t')$.
\end{lemma}
\begin{proof}
    If possible let $\exists~ t_1>t$, such that at $t_1$, $d_{CW}(r,r')$ decreases from $k'$. This implies exactly one of $r$ or, $r'$ moves in round $t_1$. First, let $r$ move while $r'$ stays still. Then $r$ must have to move clockwise to decrease the distance. Which implies $r\in \mathcal{EAS}(t_1)$ and hence $r=H(KS)$ for some movable $k$-Link $KS$. Thus in $\C(t_1)$, $k'=d_{CW}(r,r')>k$ which is a contradiction as $k'<k$. 

    So, now let us consider that $r$ stays still whereas $r'$ moves counterclockwise in $\C(t_1)$. using similar argument as in second part of Lemma~\ref{lemma: larger gaps cannot be lower than $k$} we can say that there is a movable $k$-Link $KS'$, such that $r=H(KS')$ in $\C(t_1)$, and thus $k'=d_{CW}(r,r')>k$ in $\C(t_1)$ which is again a contradiction. Thus for any $t'>t$, $d_{CW}(r,r')\ge k'$ in $\C(t')$. 
\end{proof}
\begin{lemma}
    \label{lemma: distance increase}
    Let $\mathcal{C}(t)$ be a configuration that is not distance-$k$-dispersed. Then there exists a pair of agents $r$ and $r'$ located on adjacent occupied nodes $v(r)$ and $v(r')$ such that:
\begin{enumerate}
    \item $r$ and $r'$ belong to the same $k$-Link in $\mathcal{C}(t)$, and
    \item  $d_{CW}(r, r')$ in $\mathcal{C}(t+1)>$ $d_{CW}(r, r')$ in $\mathcal{C}(t)$.
\end{enumerate}
\end{lemma}
\begin{proof}
    Since the configuration $\C(t)$ is not distance-$k$-dispersed, there exists a $k$-Link such that the number of agents on that $k$-Link is at least two. Also, by Lemma~\ref{lemma: movable link exists}, in $\C(t)$ there is at least one movable $k$-Link. Let $KS_0$ be one such movable $k$-Link. Let starting from $KS_0$, all the $k$-Links in counter clockwise direction be $KS_1, KS_2,\cdots, KS_{\alpha}$. Let $0\le x\le \alpha$ be such that $KS_x$ is the first $k$-Link  with number of agents at least two in counterclockwise direction starting from $KS_0$ (including it too). Let $r_x=H(KS_x)$ and $r_x'$ be the next agent of $r_x$ in the same $k$-Link in counterclockwise direction. Note that by definition $r_x \in \mathcal{NS}_{CW}(KS_0) \subseteq \mathcal{EAS}(t)$ but $r_x'$ is not. 
    
    Now let the missing edge $e_t$ be not incident to $v(r_y)$ in the clockwise direction for some $r_y \in \mathcal{EAS}(t)$. In this case, $r_x$ moves clockwise and $r_x'$ does not move in $\C(t)$. Thus, taking $r=r_x'$ and $r'=r_x$ is sufficient for this case. 

    Now let $e_t$ be incident to $v(r_y)$ in the clockwise direction for some $r_y \in \mathcal{EAS}(t)$. Let $r_y$ is part of the $k$-Link $KS_z$ for some $0 \le z \le \alpha$. Let $KS_z'$ be the first movable $k$-Link in the clockwise direction starting from $KS_z$. Let us consider the set $M=A \setminus \mathcal{NS}_{CCW}(KS_z')$. Now, let $r_z$ be the last agent in $M$ starting from $H(KS_z')$ in the counterclockwise direction. Then, $r_z=H(KS_{\beta})$ for some $k$-Link, with at least two agents. Now let $r_z'$ be the next agent in the counterclockwise direction of $r_z$. Now $r_z$ and $r_z'$ are in same $k$-Link by definition of $\mathcal{NS}_{CCW}(KS_z')$ and also in $\C(t)$ $r_z'$ moves counter clockwise and $r_z$ stays still. Thus $d_{CW}(r_z',r_z)$ increases in $\C(t+1)$ from $\C(t)$. So, taking $r=r_z'$ and $r'=r_z$ is sufficient for the proof. Note that for each case there exists a pair $r$ and $r'$ in the same $k$-Link such that $v(r')$ is the adjacent occupied node of $v(r)$ in clockwise direction and $d_{CW}(r,r')$ increases in $\C(t+1)$
\end{proof}
Note that in Lemma~\ref{lemma: larger gaps cannot be lower than $k$} we proved that any pair of agents on adjacent occupied nodes separated by a distance greater than $k$ can never be in the same $k$-Link further. So, the number of $k$-Links never decreases. Also in Lemma~\ref{lemma: smaller gap never decrease}, we proved that for any two agents at adjacent occupied nodes in the same $k$-Link, their smallest distance never decrease. In the worst case, let all agents be in the same $k$-Link in the initial configuration. So, to achieve Distance-$k$-Dispersion, the smallest distance between the agents for each of $O(l)$ pairs of agents in adjacent occupied nodes, has to be increased at least $k$ times. For $O(l)$ pairs, the total increase in distance should be $O(lk)$ in the worst case. Now, by Lemma~\ref{lemma: distance increase}, in each round, there exists a pair of agents in adjacent occupied nodes with distance strictly less than $k$, such that their distance increases. So, in $O(lk)$ rounds of executing \textsc{Dispersed-To-$k-$Dispersed}, the Distance-$k$-Dispersion will be solved. From this discussion, we have the following theorem.   
\begin{theorem}
   Let $\mathcal{C}(t)$ be any dispersed configuration in which all $l$ agents possess chirality at some round $t$. Then, by executing the algorithm \textsc{Dispersed-To-$k$-Dispersed}, the agents can achieve Distance-$k$-Dispersion within $O(lk)$ rounds.
\end{theorem}


\section{Conclusion}
In this work, we studied the Distance-$k$-Dispersion (D-$k$-D) problem for synchronous mobile agents in 1-interval connected rings without assuming chirality. We introduced a new technique that enables agents to simulate chirality using only local communication, bounded memory, vision, and global weak multiplicity detection and thereby showing that chirality is not a fundamental requirement for coordination in this model. With this insight, we solved D-$k$-D—and hence the standard dispersion problem—from arbitrary initial configurations, addressing an open question posed by Agarwalla et al.\ (ICDCN 2018) for even-sized rings. Finally, we provided the first finite-time algorithm that solves D-$k$-D in $O(ln)$ rounds, where $l$ is the number of agents and $n$ is the ring size.

An interesting direction for future work is to explore whether similar techniques for simulating chirality can be applied in other dynamic graph topologies or under more adversarial conditions, such as vertex permutation dynamism or weaker visibility models. Additionally, the overall time complexity of $O(ln)$ is solely due to the chirality simulation step; the subsequent algorithms for dispersion and distance-$k$-dispersion are time-optimal. This raises a natural question: What is the lower bound on the time required to simulate chirality in this model? It would be really interesting to find out the answer to this question.

\bibliography{refs}

\end{document}